\documentclass[pra,aps,nopacs,onecolumn,twoside,superscriptaddress]{revtex4}

\usepackage{soul} 
\usepackage{multirow}
\usepackage{amsmath,amsfonts,amssymb,caption,color,epsfig,graphics,graphicx,hyperref,latexsym,mathrsfs,revsymb,theorem,url,verbatim,epstopdf,enumerate}
\usepackage{amsmath,amsfonts,amssymb,caption,hyperref,color,epsfig,graphics,graphicx,latexsym,mathrsfs,revsymb,theorem,url,verbatim,epstopdf,cleveref}
\usepackage{fontenc}
\hypersetup{colorlinks,linkcolor={blue},citecolor={blue},urlcolor={red}}

\newtheorem{definition}{Definition}
\newtheorem{proposition}[definition]{Proposition}
\newtheorem{lemma}[definition]{Lemma}

\newtheorem{theorem}[definition]{Theorem}
\newtheorem{corollary}[definition]{Corollary}
\newtheorem{conjecture}[definition]{Conjecture}

\newtheorem{remark}[definition]{Remark}
\newtheorem{example}[definition]{Example}
\newtheorem{question}[definition]{Question}
\newtheorem{memo}[definition]{Memo}

\def\squareforqed{\hbox{\rlap{$\sqcap$}$\sqcup$}}
\def\qed{\ifmmode\squareforqed\else{\unskip\nobreak\hfil
\penalty50\hskip1em\null\nobreak\hfil\squareforqed
\parfillskip=0pt\finalhyphendemerits=0\endgraf}\fi}
\def\endenv{\ifmmode\;\else{\unskip\nobreak\hfil
\penalty50\hskip1em\null\nobreak\hfil\;
\parfillskip=0pt\finalhyphendemerits=0\endgraf}\fi}
\newenvironment{proof}{\noindent \textbf{{Proof.~} }}{\qed}
\def\Dbar{\leavevmode\lower.6ex\hbox to 0pt
{\hskip-.23ex\accent"16\hss}D}
\makeatletter
\def\url@leostyle{%
  \@ifundefined{selectfont}{\def\UrlFont{\sf}}{\def\UrlFont{\small\ttfamily}}}
\makeatother
\def\bcj{\begin{conjecture}}
\def\ecj{\end{conjecture}}
\def\bcr{\begin{corollary}}
\def\ecr{\end{corollary}}
\def\bd{\begin{definition}}
\def\ed{\end{definition}}
\def\bea{\begin{eqnarray}}
\def\eea{\end{eqnarray}}
\def\bem{\begin{enumerate}}
\def\eem{\end{enumerate}}
\def\bex{\begin{example}}
\def\eex{\end{example}}
\def\bim{\begin{itemize}}
\def\eim{\end{itemize}}
\def\bl{\begin{lemma}}
\def\el{\end{lemma}}
\def\bma{\begin{bmatrix}}
\def\ema{\end{bmatrix}}
\def\bpf{\begin{proof}}
\def\epf{\end{proof}}
\def\bpp{\begin{proposition}}
\def\epp{\end{proposition}}
\def\bqu{\begin{question}}
\def\equ{\end{question}}
\def\br{\begin{remark}}
\def\er{\end{remark}}
\def\bt{\begin{theorem}}
\def\et{\end{theorem}}
\def\bmm{\begin{memo}}
\def\emm{\end{memo}}

\def\btb{\begin{tabular}}
\def\etb{\end{tabular}}

\newcommand{\nc}{\newcommand}

\def\a{\alpha}
\def\b{\beta}
\def\g{\gamma}
\def\d{\delta}

\def\r{\rho}

\def\G{\Gamma}

\def\X{\Xi}

\def\S{\Sigma}

\def\O{\Omega}

 \nc{\bbA}{\mathbb{A}} \nc{\bbB}{\mathbb{B}} \nc{\bbC}{\mathbb{C}}
 \nc{\bbD}{\mathbb{D}} \nc{\bbE}{\mathbb{E}} \nc{\bbF}{\mathbb{F}}
 \nc{\bbG}{\mathbb{G}} \nc{\bbH}{\mathbb{H}} \nc{\bbI}{\mathbb{I}}
 \nc{\bbJ}{\mathbb{J}} \nc{\bbK}{\mathbb{K}} \nc{\bbL}{\mathbb{L}}
 \nc{\bbM}{\mathbb{M}} \nc{\bbN}{\mathbb{N}} \nc{\bbO}{\mathbb{O}}
 \nc{\bbP}{\mathbb{P}} \nc{\bbQ}{\mathbb{Q}} \nc{\bbR}{\mathbb{R}}
 \nc{\bbS}{\mathbb{S}} \nc{\bbT}{\mathbb{T}} \nc{\bbU}{\mathbb{U}}
 \nc{\bbV}{\mathbb{V}} \nc{\bbW}{\mathbb{W}} \nc{\bbX}{\mathbb{X}}
 \nc{\bbZ}{\mathbb{Z}}

 \nc{\bA}{{\bf A}} \nc{\bB}{{\bf B}} \nc{\bC}{{\bf C}}
 \nc{\bD}{{\bf D}} \nc{\bE}{{\bf E}} \nc{\bF}{{\bf F}}
 \nc{\bG}{{\bf G}} \nc{\bH}{{\bf H}} \nc{\bI}{{\bf I}}
 \nc{\bJ}{{\bf J}} \nc{\bK}{{\bf K}} \nc{\bL}{{\bf L}}
 \nc{\bM}{{\bf M}} \nc{\bN}{{\bf N}} \nc{\bO}{{\bf O}}
 \nc{\bP}{{\bf P}} \nc{\bQ}{{\bf Q}} \nc{\bR}{{\bf R}}
 \nc{\bS}{{\bf S}} \nc{\bT}{{\bf T}} \nc{\bU}{{\bf U}}
 \nc{\bV}{{\bf V}} \nc{\bW}{{\bf W}} \nc{\bX}{{\bf X}}
 \nc{\bZ}{{\bf Z}}

\nc{\cA}{{\cal A}} \nc{\cB}{{\cal B}} \nc{\cC}{{\cal C}}
\nc{\cD}{{\cal D}} \nc{\cE}{{\cal E}} \nc{\cF}{{\cal F}}
\nc{\cG}{{\cal G}} \nc{\cH}{{\cal H}} \nc{\cI}{{\cal I}}
\nc{\cJ}{{\cal J}} \nc{\cK}{{\cal K}} \nc{\cL}{{\cal L}}
\nc{\cM}{{\cal M}} \nc{\cN}{{\cal N}} \nc{\cO}{{\cal O}}
\nc{\cP}{{\cal P}} \nc{\cQ}{{\cal Q}} \nc{\cR}{{\cal R}}
\nc{\cS}{{\cal S}} \nc{\cT}{{\cal T}} \nc{\cU}{{\cal U}}
\nc{\cV}{{\cal V}} \nc{\cW}{{\cal W}} \nc{\cX}{{\cal X}}
\nc{\cZ}{{\cal Z}}

\nc{\hA}{{\hat{A}}} \nc{\hB}{{\hat{B}}} \nc{\hC}{{\hat{C}}}
\nc{\hD}{{\hat{D}}} \nc{\hE}{{\hat{E}}} \nc{\hF}{{\hat{F}}}
\nc{\hG}{{\hat{G}}} \nc{\hH}{{\hat{H}}} \nc{\hI}{{\hat{I}}}
\nc{\hJ}{{\hat{J}}} \nc{\hK}{{\hat{K}}} \nc{\hL}{{\hat{L}}}
\nc{\hM}{{\hat{M}}} \nc{\hN}{{\hat{N}}} \nc{\hO}{{\hat{O}}}
\nc{\hP}{{\hat{P}}} \nc{\hR}{{\hat{R}}} \nc{\hS}{{\hat{S}}}
\nc{\hT}{{\hat{T}}} \nc{\hU}{{\hat{U}}} \nc{\hV}{{\hat{V}}}
\nc{\hW}{{\hat{W}}} \nc{\hX}{{\hat{X}}} \nc{\hZ}{{\hat{Z}}}

\nc{\hn}{{\hat{n}}}

\def\diag{\mathop{\rm diag}}
\def\dim{\mathop{\rm Dim}}

\def\I{\mathop{\rm I}}

\def\lin{\mathop{\rm span}}

\def\max{\mathop{\rm max}}
\def\min{\mathop{\rm min}}

\def\rank{\mathop{\rm rank}}

\def\ox{\otimes}

\def\ra{\rightarrow}

\newcommand{\ket}[1]{|#1\rangle}
\newcommand{\proj}[1]{| #1\rangle\!\langle #1 |}
\newcommand{\ketbra}[2]{|#1\rangle\!\langle#2|}

\def\Dbar{\leavevmode\lower.6ex\hbox to 0pt
{\hskip-.23ex\accent"16\hss}D}
\begin{document}

\title{Equality condition for a matrix inequality by partial transpose}

\date{\today}


\author{Nalan Wang}\email[]{nalanwang@buaa.edu.cn}
\affiliation{LMIB(Beihang University), Ministry of Education, and School of Mathematical Sciences, Beihang University, Beijing 100191, China}

\author{Lin Chen}\email[]{linchen@buaa.edu.cn (corresponding author)}
\affiliation{LMIB(Beihang University), Ministry of Education, and School of Mathematical Sciences, Beihang University, Beijing 100191, China}

\begin{abstract}
The partial transpose map is a linear map widely used quantum information theory. We study the equality condition for a matrix inequality generated by partial transpose, namely $\rank(\sum^K_{j=1} A_j^T \otimes B_j)\le K \cdot \rank(\sum^K_{j=1} A_j \otimes B_j)$, where $A_j$'s and $B_j$'s are respectively the matrices of the same size, and $K$ is the Schmidt rank. We explicitly construct the condition when $A_i$'s are column or row vectors, or $2\times 2$ matrices. For the case where the Schmidt rank equals the dimension of $A_j$, we extend the results from $2\times 2$ matrices to square matrices, and further to rectangular matrices. In detail, we show that $\sum^K_{j=1} A_j \otimes B_j$ is locally equivalent to an elegant block-diagonal form consisting solely of identity and zero matrices. We also study the general case for $K=2$, and it turns out that the key is to characterize the expression of matrices $A_j$'s and $B_j$'s. 
\end{abstract}

\maketitle


%\tableofcontents

%\Large

\section{Introduction}

Composite systems are a basic physical setting in quantum mechanics and information. For example, nonlocality and entanglement occurs between two or more particles. They form a composite bipartite or multipartite system and entanglement detection has attracted lots of interests over the past decades \cite{hhh96,horodecki1997,hh1999}. The quantification of entanglement characterizes how much entanglement a multipartite state has, and is a basic problem in quantum information theory and application. Well-known entanglement measures have been constructed such as distillable entanglement \cite{Divincenzo2000Evidence,Chen2008Rank,Chen2012Distillability,Lin2016Non,QIAN2021139} and Schmidt measure in terms of Schmidt rank of a multipartite pure state
\cite{Eisert2000The}. They work for both bipartite and multipartite systems, and the latter have a more complex structure. For example, the constraints over the entropies of marginals of a multipartite state can be expressed as inequalities 
\cite{2005A,2012Infinitely}. Recently the $0$-Renyi entropy, being the limit of $\a$-Renyi entropy with $\a\ra0$, has been applied to estimate the relation among three bipartite marginals in Ref. \cite{cadney2014}. The reference proposed a conjectured matrix inequality $\rank(\sum^K_{j=1} R_j^T \otimes S_j)\le K \cdot \rank (\sum^K_{j=1} R_j \otimes S_j)$, in terms of matrix Kronecker product and Schmidt rank $K$, which shows potentially undiscovered multipartite quantum correlations. The conjecture has been proven via a pure algebraically style in \cite{WOS:000966945300001}. However, the condition saturating the inequality is unknown yet except some primary results in \cite{tan}, due to the involved sum of Kronecker products of matrices in the above inequality.

In this paper, we shall explore more cases for the condition. First, we lay the groundwork for the proofs as follows. Lemmas \ref{lemma1}-\ref{mtimesnA&B} provide explanations and clarifications on matrix rank, linear dependence, and image spaces. Definitions \ref{definition4} and \ref{def:equivalent} formalize the Kronecker product and its properties. The inequality of Theorem \ref{cj:1} has already been established in prior work. Lemmas \ref{le:X=A+yB}-\ref{eq:K<=} extend properties of matrix rank and image spaces, all of which will be utilized in subsequent proofs. In Theorem \ref{le:ANYm2,K>2}, we present the necessary and sufficient conditions for the inequality to become an equality when $R_i$ is a column vector. Here, the matrix has equal rank under transposition of its $R_j$ and $S_j$ part. Through an analysis from specific instances to general cases, we deduce that $\rank(\sum^K_{j=1} R_j \otimes S_j^T)= K \cdot \rank (\sum^K_{j=1} R_j \otimes S_j)$ holds if and only if $m_1n_2\geq m_2$ and there exists an invertible matrix $Q\in \bbM_{n_2}$ such that $\sum^K_{j=1} R_j \otimes S_j$ can be reduced to a matrix composed of $k_id_i$ linearly independent vectors and zero matrices. A faster verification method for this theorem is to analyze the condition of $K = 2 $, which is the initial case that lead us to formulate Theorem \ref{le:ANYm2,K>2}, too. As is mentioned in the paper, this theorem can also be extended to the discussion of row vectors. 
 
Next, we proceed to investigate the condition when  $R_j$ in $\rank(\sum^K_{j=1} R_j \otimes S_j^T)= K \cdot \rank (\sum^K_{j=1} R_j \otimes S_j)$ is a $2\times2$ matrix. We respectively obtain Lemmas \ref{le:SR=2,Ri=Si=2X2}-\ref{le:sr=4}. Lemma \ref{le:SR=2,Ri=Si=2X2} is the result of the condition when $S_1,S_2\in\bbM_2$. The following Lemma \ref{le:SR=3+4,Ri=Si=2X2} states two conditions when the Schmidt rank of order-four matrix $X$ has rank three and four. Then we explore the condition when the Schmidt rank is two in Lemma \ref{le:SR=2,Ri=2X2} and the Schmidt rank is three in Lemma \ref{le:sr=4} in which we both have $S_1,S_2\in\bbM_{m_2,n_2}$. Theorem \ref{le:m1=n1=2,sr=4}, presented next, demonstrates one of the main results in this paper. We conclude that if the Schmidt rank matrix $X$ has rank four, then the inequality is saturated if and only if $X$ has rank $r$ where $r=\rank X_i$, $X_i,i=1,2,3,4$ is a $m_2\times n_2$ matrix and $X$ can be reduced to the simplest form $\left[\begin{matrix}\left[\begin{matrix}I_r&O\\O&O\\\end{matrix}\right]&\left[\begin{matrix}O&I_r\\O&O\\\end{matrix}\right]\\\left[\begin{matrix}O&O\\I_r&O\\\end{matrix}\right]&\left[\begin{matrix}O&O\\O&I_r\\\end{matrix}\right]\\\end{matrix}\right]$. Moreover, in Theorem \ref{k^2Schmidt_rank} we generalize the case of Schmidt rank $2\times 2$ to $k\times k$ where the necessary and sufficient condition is that $X$ can be transformed into the general form in \eqref{XEij}, which is a beautiful form. Finally, we extend this concise form to the general case of $m_1 \times n_1$  matrices $R_i$, which yields Lemma \ref{m1n1Schmidt_rank}. An illustrative example \eqref{X2times3} is provided along with corresponding explanations and verification.

Third, we simplify every Schmidt-rank-two matrix to \eqref{SVD_result:EK} using SVD. We repeatedly use this way to investigate $\rank X$, reducing $X$ it to a quasi-diagonal form in Lemma \ref{SVDX}, with most block matrices taking on a standard form. The only imperfection is that there always remains an irreducible matrix in the upper-right corner, which depends on the specific form of the original matrix. In fact, SVD serves as an effective analytical tool. In addition, Lemma \ref{le:X=AotimesB+CotimesD}, Corollary \ref{cr:X=AotimesB+CotimesD} and Lemma \ref{le:X=AotimesB+CotimesD=square} can be some new ideas on the investigation of Schmidt-rank-two matrices. Lemma \ref{le:blockDiagonal} presents the key finding that facilitates a deeper understanding of this paper. These results are highly valuable for deriving general principles.

The Schmidt rank (i.e., tensor rank) has been widely used in other quantum-information topics, such as graph states \cite{WOS:000222471400029}, geometric measure of entanglement \cite{WOS:000187004700026},
resource theory \cite{WOS:000463910100001}, multipartite state transformation under local unitary equivalence \cite{WOS:000274002900006} and SLOCC equivalence \cite{WOS:000259793800006}. Our mathematical results may also shed light on the existing challenges towards these topics.

The rest of this paper is organized as follows. In Sec. \ref{sec:pre} we introduce the preliminary facts used in this paper. Then we investigate the equality condition for the inequality in Sec. \ref{sec:condition=equality} and \ref{sec:2x2}, when the matrices $A_j$'s are respectively column/row vectors or two-by-two matrices. In Sec. \ref{sec:rank=bm} we study the bipartite matrices with Schmidt rank two. Finally we conclude in Sec. \ref{sec:con}.

\section{Preliminaries}
\label{sec:pre}

In this section we introduce the preliminary knowledge of this paper. We denote $\bbM_{m,n}$ as the set of $m\times n$ matrices over the field of complex numbers, and $\bbM_n:=\bbM_{n,n}$ for convenience. Let $I_{n}$ be the $n\times n$ identity matrix. The following facts are known in linear algebra.

\begin{lemma}
	\label{lemma1}	
	The following inequalities
	\begin{eqnarray}
	\label{baseineq1}
	r(A_1)\le r(\bma A_1&A_2&\cdots&A_n
	\ema)\le r(A_1)+\cdots+r(A_n),
	\end{eqnarray}
	\begin{eqnarray}
	\label{baseineq2}
	r(A)+r(C)\le r(\bma A&0\\
	B&C\ema), 
	\end{eqnarray}
   hold for any  block matrix.
   \qed
\end{lemma}
\begin{lemma}
	\label{lemmaadd}
	(i) Suppose $A_1,\cdots,A_n\in \bbM_{m_2,n_2}$ are linearly independent, $R\in \bbM_{n_2}$ is invertible. Then $A_1\cdot R,\cdots,A_n\cdot R$ are linearly independent.
	
	(ii) Suppose $A_1,\cdots,A_n,A_{n+1}\in \bbM_{m_2,n_2}$, and $R\in \bbM_{n_2}$ is invertible. If $A_{n+1}\in \lin \{A_1,\cdots,A_n\}$, then $A_{n+1}\cdot R\in \lin \{A_1\cdot R,\cdots,A_n\cdot R\}$.
	\qed
\end{lemma}

Here is another fact on the linear independence.
\begin{lemma}
\label{mtimesnA&B}
Let $A$ and $B$ be two $m\times n$ matrices with $m\ge n$. If the span of $A$ and $B$ has no full-rank matrix, then there exist $s$ columns $a_1,...,a_s$ of $A$, and $s$ columns $b_1,...,b_s$ of $B$ in the same position as those of $A$, such that $a_1,...,a_s,b_1,...,b_s$ belong to a subspace of $\bbC^m$ of dimension $s-1$.
\end{lemma}
\begin{proof}
We prove the statement by contradiction. Firstly, if the span of $A$ and $B$ has no full-rank matrix, for any $\alpha,\beta \in \bbR$, the column vectors of matrix $\alpha 
 A+\beta B$ are linearly dependent. Secondly, we consider any collection of column vectors of $A$ and $B$, say $\rank(a_{i1},a_{i2},...,a_{is},b_{i1},b_{i2},...,b_{is})\geq s$. If $s=n$, then the columns of matrix $\alpha  A+\beta B$ are linearly independent. This leads to a contradiction with our hypothesis. 
\end{proof}

Next, a bipartite matrix $M\in \bbM_{m_1,n_1}\ox\bbM_{m_2,n_2}$ is denoted as $M=\sum^{m_1}_{i=1}\sum^{n_1}_{j=1}\ketbra{i}{j}\otimes M_{i,j}$ with $M_{i,j}\in \bbM_{m_2,n_2}.$  We denote the partial transposes of $M$ w.r.t. system $A$ and $B$ as $M^{\Gamma_A}=\sum^{m_1}_{i=1}\sum^{n_1}_{j=1}\ketbra{j}{i}\otimes M_{i,j}$, and $M^{\Gamma_B}=\sum^{m_1}_{i=1}\sum^{n_1}_{j=1}\ketbra{i}{j}\otimes M_{i,j}^T$, respectively. Further, the number of linearly independent blocks $M_{i,j}$ is called the Schmidt rank $Sr:=Sr(M)$ of $M$. One can show that
$Sr(M)
=
Sr(M^{\G_A})=Sr(M^{\G_B})
=Sr(M^T).	
$

\begin{definition}
\label{definition4}
The Schmidt-rank-$k$ bipartite matrix $X\in\bbM_{m_1,n_1}\otimes\bbM_{m_2,n_2}$ equals the sum of $A_j  \otimes B_j$ from $j=1$ to $k$ where $A_j$ is an $m_1 \times n_1$ matrix and $B_j$ is an $m_2 \times n_2$ matrix. 
%We also know that  $\{X_1,X_2,…X_k\}$ are linearly independent if and only if the sum of $k_i X_i$ from $i=1$ to $k$ equals zero only when $k_i$ all be zero. 
\qed
\end{definition}

To classify bipartite matrices with a given Schmidt rank, we review the local equivalence of two bipartite matrices.

\begin{definition}
\label{def:equivalent}
	Two matrices $M$ and $N$ are said to be locally equivalent, $M\sim N$, if there exist invertible tensor product matrices $P_1\otimes P_2$ and $Q_1\otimes Q_2$ such that $(P_1\otimes P_2)M(Q_1\otimes Q_2)=N$. 
	\qed
\end{definition}
Here  $M=\sum^{m_1}_{i=1}\sum^{n_1}_{j=1}\ketbra{i}{j}\otimes M_{i,j}$, $P_1$ and $Q_1$ correspond to block-row and block-column operations on $M$ respectively, $P_2$ and $Q_2$ correspond to row and column operations on each block $M_{i,j}$ respectively, which is $M_{i,j}\rightarrow Q_1M_{i,j}Q_2$. So $M$ and $N$ have the same rank and Schmidt rank. One can similarly show that $M^{\G_B}$ and $N^{\G_B}$ have the same rank and Schmidt rank. Next the following fact is from the paper \cite{WOS:000966945300001}.

\begin{lemma}
\label{cj:1}
For any $M\in \bbM_{m_1,n_1}\ox\bbM_{m_2,n_2}$, we have 
\bea
\label{eq:cj2}
r(M^{\G_B})\le Sr(M)\cdot r(M).
\eea
\end{lemma}
The fact is equivalent to $r(M^{\G_A})\le Sr(M)\cdot r(M)$. If we write $M=\sum^K_{i=1} P_i \ox Q_i$, where $P_1,\dots,P_K\in \bbM_{m_1,n_1}$ are linearly independent, and
$Q_1,\dots,Q_K\in \bbM_{m_2,n_2}$ are linearly independent, then it is also equivalent to
$ r \bigg( \sum^K_{i=1} P_i \ox Q_i^T \bigg)
\le
K\cdot r \bigg( \sum^K_{i=1} P_i \ox Q_i \bigg)$.

\begin{lemma}
\label{le:X=A+yB}
Suppose there exist $x\in \bbC$ and $y\in \bbC$ such that $xA+yB$ is full-rank where $A$ and $B$ are both order-n matrices. The number of linearly independent solutions $(a,b)$ making the matrix $aA+bB$ rank-deficient does not exceed $n$.
\end{lemma}

\begin{proof}
We assume that $\det(aA+bB)$ is a polynomial in two variables, $a$ and $b$. Respectively, let $P(a,b)=\det(aA+bB)$, in which $a\in \bbC$ and $b\in \bbC$. 

If there exist $x\in \bbC$ and $y\in \bbC$ such that $\det(xA+yB)\neq 0$, then $P(a,b)$ is a zero polynomial. Also, if $(a,b)$ is a solution of the equation, $k(a,b)$ is another solution of this equation where $k\in \bbC$. Therefore, we can regard each solution as a direction, and different solutions are distinct in the sense of linear independence.  

Thus, an analogy can be employed. Considering an order-n matrix $A$, if $\det(xA)\neq 0$, the number of solutions $a$ that make matrix $\det(aA)=0$ does not exceed $n$ because $\det(aA)=0$ is a polynomial of degree $n$. So the set of zeros that satisfy the equation $\det(aA+bB)=0$ consists of at most $n$ distinct points in the complex projective space, with each zero corresponding to a direction $(a,b)$, which means the number of linearly independent solutions $(a,b)$ that make matrix $(aA+bB)$ rank-deficient does not exceed $n$.
\end{proof}

To conclude this section, we propose two primary facts from linear algebra.
We shall refer to $R(X)$ as the range subspace of matrix $X$.

\begin{lemma}
\label{le:rankA=rankA1}
Suppose $A=\left[\begin{matrix}A_1&A_2\\A_3&A_4\\\end{matrix}\right]\in \bbM_{a,b}$ where $A_1\in \bbM_{c,d},A_2\in \bbM_{c,b-d},A_3\in \bbM_{a-c,d},A_4\in \bbM_{a-c,b-d}$. Then $\rank A=\rank A_1$ if and only if
\begin{align}
R(A)=R\begin{bmatrix}A_1\\A_3\\\end{bmatrix}\ and\ R(A_3^T)\subseteq R(A_1^T).
\end{align}
\end{lemma}

\begin{proof}
(i) If $R(A)=R\begin{bmatrix}A_1\\A_3\\\end{bmatrix}$, each column of $A$ is a linear combination of the columns of $\begin{bmatrix}A_1\\A_3\\\end{bmatrix}$, which means each column of $\begin{bmatrix}A_2\\A_4\\\end{bmatrix}$ is a linear combination of the columns of $\begin{bmatrix}A_1\\A_3\\\end{bmatrix}$. Therefore there exists an invertible matrix $P_1$ of order $b$ such that $\left[\begin{matrix}A_1&A_2\\A_3&A_4\\\end{matrix}\right]P_1=\left[\begin{matrix}A_1&O\\A_3&O\\\end{matrix}\right]$. Then due to $R(A_3^T)\subseteq R(A_1^T)$, each column of $A_3^T$ is a linear combination of the columns of $A_1^T$. So there exists an invertible matrix $P_2$ of order $a$ such that $P_2\left[\begin{matrix}A_1&A_2\\A_3&A_4\\\end{matrix}\right]P_1=\left[\begin{matrix}A_1&O\\O&O\\\end{matrix}\right]$. In conclusion, $\rank A=\rank A_1$;

(ii) Conversely, if $\rank A=\rank A_1$, there exist invertible matrices $P_1$ of order $b$ and $P_2$ of order $a$ such that $P_2\left[\begin{matrix}A_1&A_2\\A_3&A_4\\\end{matrix}\right]P_1=\left[\begin{matrix}A_1&O\\O&O\\\end{matrix}\right]$. Thus each column of $\begin{bmatrix}A_2\\A_4\\\end{bmatrix}$ is a linear combination of the columns of $\begin{bmatrix}A_1\\A_3\\\end{bmatrix}$ and each row of $A_3$ is a linear combination of the columns of $A_1$, which means 
\begin{align*}
    R(A)=R\begin{bmatrix}A_1\\A_3\\\end{bmatrix}\ and\ R(A_3^T)\subseteq R(A_1^T).
\end{align*}
\end{proof}

\begin{lemma}
\label{le:rankA=rankB+rankC}
Suppose matrix $A=\left[\begin{matrix}B&C\\\end{matrix}\right]$. Then $\rank A= \rank B+\rank C$ if and only if $R(A)=R(B)+R(C)$. 
\end{lemma}
\begin{proof}
(i) Firstly, if $R(A)=R(B)+R(C)$, the image space of $A$ is the direct sum of the image spaces of $B$ and $C$. Obviously we obtain $\rank A= \rank B+\rank C$;

(ii) Secondly, based on rank properties of matrices, for matrix $A=\left[\begin{matrix}B&C\\\end{matrix}\right]$ we have $\rank A\geq \rank B,\rank A\geq \rank C$. Thus when $\rank A= \rank B+\rank C$, respectively, neither the columns of $C$ lie in the span of columns of $B$, nor do the columns of $B$ lie in the span of columns of $C$. Therefore we obtain $R(A)=R(B)+R(C)$.
\end{proof}

\begin{lemma}
\label{le:xA+yB=singular}
Let $X=\sum_{i=1}^{k}{A_i\otimes B_i}\in\ \bbC^m\otimes \bbC^n$ where $A_i$ are linear independent matrices in $\bbC^m$ and $B_i$ are linear independent matrices in $\bbC^n$. Hence $SrX=k$. We have $X^\G=\sum_{i=1}^{k}{A_i^\G\otimes B_i}$, one can verify that $\rank X=\rank X^\G$ that can be demonstrated in the case that $k=1$.
\end{lemma}

\begin{lemma}
\label{eq:K<=}
For any $M\in \bbM_{m_1,n_1}\ox\bbM_{m_2,n_2}$, suppose the following equation holds:
\bea
\label{eq:cj2=3}
r(M^{\G_B})= Sr(M)\cdot r(M), 
\eea
in addition, it is equivalent to the following equation holding:
\bea
\label{eq:cj2=4}
r \bigg( \sum^K_{i=1} R_i \ox S_i^T \bigg) = K\cdot r \bigg( \sum^K_{i=1} R_i \ox S_i \bigg), 
\eea
then we can derive the solution: 
$r \bigg( \sum^K_{i=1} R_i \ox S_i^T \bigg)$ can only take $1, 2, ...\min{(m_1n_2,m_2n_1)}$, which means for the positive integer $K$, the following inequality holds:
\bea
\label{eq:K<=m1n2,m2n1}
K\le \min{(m_1n_2,m_2n_1)}. 
\eea
\end{lemma}

\section{Condition for the equality: $n_1=1$}
\label{sec:condition=equality}

In this section, we investigate the matrix $M\in \bbM_{m_1,n_1}\ox\bbM_{m_2,n_2}$ for which the inequality in \eqref{cj:1} is saturated with $n_1=1$ in Lemma \ref{le:ANYm2,K>2}. Further we introduce the case of $m_j=n_j=2$ in Lemmas \ref{le:SR=2,Ri=Si=2X2} and \ref{le:SR=3+4,Ri=Si=2X2}. We shall assume $K>1$, because the inequality is always saturated for $K=1$.

We consider the column vectors $R_i\in \bbC^{m_1}$ and matrices $S_i\in \bbM_{m_2\times{n_2}}$ where $R_i, S_i$ are linear independent in the equality $ r \bigg( \sum^K_{i=1} R_i \ox S_i^T \bigg)
=
K\cdot r \bigg( \sum^K_{i=1} R_i \ox S_i \bigg)$, for $K>1$. Here, we need $m_1n_2\geq m_2$. Then the equality holds if the following condition is satisfied.

Let $A$ be the set of all integers dividing the integer $n$ satisfying $2\le n\le \min{\{m_1,m_2\}}$. We assume that
$$A=\{k_1,k_2,...,k_r\}, $$
in which the cardinality of set $A$ is $r$, and $2\le k_1< ... <k_{r-1} < k_r= \min{\{m_1,m_2\}}$. Using Lemma \ref{eq:K<=}, $K\le m_2 n_1=m_2$ can be any element of $A$. Then when $K=k_i, i=1,2,..r$, $\rank \begin{bmatrix}S_1\\S_2\\\vdots\\S_K\\\end{bmatrix}=d_{ij},j\in \bbN^*, k_id_{ij}\le m_2$ and $\rank \begin{bmatrix}S_1&S_2&\ldots&S_K\\\end{bmatrix}=k_id_{ij}$. 

Or more concisely, the above theorem can be restated in a simpler form as follows.

There exists an invertible matrix $Q\in \bbM_{n_2}$ such that 

\bea
\begin{bmatrix}S_1\\S_2\\\vdots\\S_K\\\end{bmatrix}Q=\begin{bmatrix}\vec{u_{11}}&\vec{u_{12}}&\ldots&\vec{u_{1d_{i}}}&0&\ldots&0\\\vec{u_{21}}&\vec{u_{22}}&\ldots&\vec{u_{2d_{i}}}&0&\ldots&0\\\vdots&\vdots&\ldots&\vdots&\vdots&\ldots&\vdots\\\vec{u_{k_i1}}&\vec{u_{k_i2}}&\ldots&\vec{u_{k_id_{i}}}&0&\ldots&0\\\end{bmatrix}, 
\eea

where $\vec u_{11},\vec{u_{12}},...\vec{u_{1d_{i}}},\vec{u_{21}},\vec{u_{22}},...\vec{u_{2d_{i}}},\vec{u_{k_i1}},\vec{u_{k_i2}},...\vec{u_{k_id_{i}}}\in \bbC^{m_2}$ are linearly independent column vectors such that $k_id_{ij}\times m_2$ matrix $\left[\begin{matrix}{\vec{u_{11}}}^T\\{\vec{u_{12}}}^T\\\vdots\\{\vec{u_{1d_{i}}}}^T\\{\vec{u_{21}}}^T\\\vdots\\\vec{u_{k_id_{i}}}^T\\\end{matrix}\right]$ is a matrix of full rank and for $K=k_i$, each row of $\left[\begin{matrix}\vec{u_{11}}&\vec{u_{12}}&\ldots&\vec{u_{1d_{i}}}&0&\ldots&0\\\vec{u_{21}}&\vec{u_{22}}&\ldots&\vec{u_{2d_{i}}}&0&\ldots&0\\\vdots&\vdots&\ldots&\vdots&\vdots&\ldots&\vdots\\\vec{u_{k_i1}}&\vec{u_{k_i2}}&\ldots&\vec{u_{k_id_{i}}}&0&\ldots&0\\\end{matrix}\right]$ is obtained by performing elementary column transformations on the column vectors of matrices $S_1,S_2,...S_K$. 

\begin{lemma}
\label{le:ANYm2,K>2}
Suppose the column vectors $R_i\in \bbC^{m_1}$ and matrices $S_i\in \bbM_{m_2\times{n_2}}$ where $R_i, S_i$ are linearly independent in the equality $ r \bigg( \sum^K_{i=1} R_i \ox S_i^T \bigg)
=
K\cdot r \bigg( \sum^K_{i=1} R_i \ox S_i \bigg)$. Then the equality holds if and only if the following condition is satisfied:

$m_1n_2\geq m_2$ and there exists an invertible matrix $Q\in \bbM_{n_2}$ such that 

\bea
\left[\begin{matrix}S_1\\S_2\\\vdots\\S_K\\\end{matrix}\right]Q=\left[\begin{matrix}\vec{u_{11}}&\vec{u_{12}}&\ldots&\vec{u_{1d_{i}}}&0&\ldots&0\\\vec{u_{21}}&\vec{u_{22}}&\ldots&\vec{u_{2d_{i}}}&0&\ldots&0\\\vdots&\vdots&\ldots&\vdots&\vdots&\ldots&\vdots\\\vec{u_{k_i1}}&\vec{u_{k_i2}}&\ldots&\vec{u_{k_id_{i}}}&0&\ldots&0\\\end{matrix}\right], 
\eea

where $u_{11},\vec{u_{12}},...\vec{u_{1d_{i}}},\vec{u_{21}},\vec{u_{22}},...\vec{u_{2d_{i}}},\vec{u_{k_i1}},\vec{u_{k_i2}},...\vec{u_{k_id_{i}}}\in \bbC^{m_2}$ are linearly independent column vectors such that $k_id_{i}\times m_2$ matrix $\left[\begin{matrix}{\vec{u_{11}}}^T\\{\vec{u_{12}}}^T\\\vdots\\{\vec{u_{1d_{i}}}}^T\\{\vec{u_{21}}}^T\\\vdots\\{u_{k_id_{i}}}^T\\\end{matrix}\right]$ is a matrix of full rank.

\end{lemma}

\begin{proof}
When $r \bigg( \sum^K_{i=1} R_i \ox S_i^T \bigg)=1$, the theorem holds obviously. Based on the above definition, we now discuss the case where $r \bigg( \sum^K_{i=1} R_i \ox S_i^T \bigg)=n$ and $2\le n\le \min{\{m_1,m_2\}}$. Let $A$ be the set of all integers dividing the integer $n$ satisfying $2\le n\le \min{\{m_1,m_2\}}$. We assume that
$$A=\{k_1,k_2,...,k_r\}, $$
in which the cardinality of set $A$ is $r$, and $2\le k_1< ... <k_{r-1} < k_r= \min{\{m_1,m_2\}}$. Using Lemma \ref{eq:K<=}, $K\le m_2 n_1=m_2$ can be any element of $A$. Then when $K=k_i, i=1,2,..r$, firstly, we discuss $r \bigg( \sum^K_{i=1} R_i \ox S_i \bigg)$. For the matrix $\bigg( \sum^K_{i=1} R_i \ox S_i \bigg)$, there exists an invertible matrix $P$ that is order-$m_1$ such that $PR_j=\vec{e_j}\in \bbC^{m_1}$ where $\vec{e_j}$ is an an orthonormal basis in $\bbC^{m_1}$ because $R_j, j=1,2,...k_i$ are linear independent, $i=1,2,...k_i$. From the properties of the Kronecker product of matrices, we can derive that $r \bigg( \sum^K_{i=1} R_i \ox S_i \bigg)=r ((P\ox I_{m_2})\bigg( \sum^K_{i=1} R_i \ox S_i \bigg)(1\ox I_{n_2}))=r ((PR_1 )\ox (I_{m_2}S_1I_{n_2})+(PR_2  )\ox (I_{m_2}S_2I_{n_2})+...+(PR_K )\ox (I_{m_2}S_KI_{n_2}))=r\begin{bmatrix}S_1\\S_2\\\vdots\\S_K\\O\\\end{bmatrix}=r\begin{bmatrix}S_1\\S_2\\\vdots\\S_K\\\end{bmatrix}$, which means the rank of the $k_im_2\times n_2$ matrix $r\begin{bmatrix}S_1\\S_2\\\vdots\\S_K\\\end{bmatrix}=d_i$ where $k_id_i=n$. 

Secondly, we discuss $r \bigg( \sum^K_{i=1} R_i \ox S_i^T \bigg)$. Respectively, for the matrix $\bigg( \sum^K_{i=1} R_i \ox S_i^T \bigg)$, we have the invertible matrix $P$ that is order-$m_1$ such that $r \bigg( \sum^K_{i=1} R_i \ox S_i^T \bigg)=r ((P\ox I_{n_2})(R_1 \ox S_1^T + R_2 \ox S_2^T +...+  R_K \ox S_K^T)(1\ox I_{m_2}))=r ((PR_1 )\ox (I_{n_2}S_1^TI_{m_2})+(PR_2  )\ox (I_{n_2}S_2^TI_{m_2})+...+(PR_2  )\ox (I_{n_2}S_K^TI_{m_2}))=r\begin{bmatrix}S_1^T\\S_2^T\\\vdots\\S_K^T\\O\\\end{bmatrix}=r\begin{bmatrix}S_1^T\\S_2^T\\\vdots\\S_K^T\\\end{bmatrix}$, which means the rank of the $k_in_2\times m_2$ matrix $r\begin{bmatrix}S_1^T\\S_2^T\\\vdots\\S_K^T\\\end{bmatrix}=k_id_i$ where $k_id_i=n$. 

Besides, $S_i, i=1,2,...K$ are linear independent, which means $\rank S_1 = \rank S_2 = ...=\rank S_K=d$. Respectively, there exists an invertible matrix $Q\in \bbM_{n_2}$ such that 

\bea
\left[\begin{matrix}S_1\\S_2\\\vdots\\S_K\\\end{matrix}\right]Q=\left[\begin{matrix}\vec{u_{11}}&\vec{u_{12}}&\ldots&\vec{u_{1d_{i}}}&0&\ldots&0\\\vec{u_{21}}&\vec{u_{22}}&\ldots&\vec{u_{2d_{i}}}&0&\ldots&0\\\vdots&\vdots&\ldots&\vdots&\vdots&\ldots&\vdots\\\vec{u_{k_i1}}&\vec{u_{k_i2}}&\ldots&\vec{u_{k_id_{i}}}&0&\ldots&0\\\end{matrix}\right], 
\eea

where $u_{11},\vec{u_{12}},...\vec{u_{1d_{i}}},\vec{u_{21}},\vec{u_{22}},...\vec{u_{2d_{i}}},\vec{u_{k_i1}},\vec{u_{k_i2}},...\vec{u_{k_id_{i}}}\in \bbC^{m_2}$ are linearly independent column vectors such that $k_id_{i}\times m_2$ matrix $\left[\begin{matrix}{\vec{u_{11}}}^T\\{\vec{u_{12}}}^T\\\vdots\\{\vec{u_{1d_{i}}}}^T\\{\vec{u_{21}}}^T\\\vdots\\{u_{k_id_{i}}}^T\\\end{matrix}\right]$ is a matrix of full rank. \\

In particular, we have that the equality holds if the following condition is satisfied:

$K=2$, $\rank\begin{bmatrix}S_1\\S_2\\\end{bmatrix}=d$ and $\rank\begin{bmatrix}S_1&S_2\\\end{bmatrix}=2d$ where $d=1,2,...2m$ and $2m$ is the largest even number not exceeding $m_2$. Respectively, there exists an invertible matrix $Q\in \bbM_{n_2}$ such that 

\bea
\left[\begin{matrix}S_1\\S_2\\\end{matrix}\right]Q=\left[\begin{matrix}\vec{u_1}&\vec{u_2}&\ldots&\vec{u_d}&0&\ldots&0\\\vec{v_1}&\vec{v_2}&\ldots&\vec{v_d}&0&\ldots&0\\\end{matrix}\right], 
\eea

where $\vec{u_1},\vec{u_2},...\vec{u_d},\vec{v_1},\vec{v_2},...\vec{v_d}\in \bbC^{m_2}$ are linearly independent column vectors such that $\left[\begin{matrix}{\vec{u_1}}^T\\{\vec{u_2}}^T\\\vdots\\{\vec{u_d}}^T\\{\vec{v_1}}^T\\\vdots\\{\vec{v_d}}^T\\\end{matrix}\right]$ is a matrix of full rank. These column vectors $\vec{u_1},\vec{u_2},...\vec{u_d},\vec{v_1}\in \bbC^{m_2}$ are precisely the columns of the new matrix obtained after performing column transformations of $Q$ on the original matrix $S_1$ and $\vec{v_1},\vec{v_2},...\vec{v_d}\in \bbC^{m_2}$ are precisely the columns of the new matrix obtained after performing column transformations of $Q$ on the original matrix $S_2$. 

Furthermore, the necessity is clear.
\end{proof}
On the other hand, if $R_i$ is a row vector, then one can perform the transpose on the matrix $\sum_j R_j \otimes S_j$ and $\sum_j R_j \otimes S_j^T$ at the same time. One can obtain the equality condition using Theorem \ref{le:ANYm2,K>2}.
A similar discussion applies when $S_i$ is a row or column vector due to the switch of $R_i$ and $S_i$. Hence, we have finished the study over the equality condition of \eqref{eq:cj2} when $R_i$ or $S_i$ has rank one. In the next subsection, we shall investigate more complex cases. The simplest one occurs when $R_i$ and $S_i$ are both order-two.

\section{Condition for the equality: $m_1=n_1=2$}
\label{sec:2x2}

In this section, we investigate the equality condition with $2\times2$ matrices $R_j$'s in Theorem \ref{cj:1}. We firstly study the condition with $2\times2$ matrices $S_i$'s and Schmidt-rank-two matrix $X$ in Lemma \ref{le:SR=2,Ri=Si=2X2}. Then we show in Lemma \ref{le:SR=3+4,Ri=Si=2X2} that, if the Schmidt rank of order-four matrix $X$ has rank three, then the inequality in Theorem \ref{cj:1} is strict. On the other hand if the Schmidt rank of order-four matrix $X$ has rank four, then the inequality in Theorem \ref{cj:1} is saturated if and only if $X$ has rank one. Then we extend these facts to matrices $S_i$'s of any sizes in Lemmas \ref{le:SR=2,Ri=2X2}-\ref{le:m1=n1=2,sr=4}.

\begin{lemma}
\label{le:SR=2,Ri=Si=2X2}
Let $R_1,R_2\in\bbM_2$ be linearly independent, and $S_1,S_2\in\bbM_2$ be also linearly independent in  the following equation.
\begin{eqnarray}
\label{eq:SR=2,Ri=Si=2X2}
\rank(R_1\otimes S_1^T+R_2\otimes S_2^T)=2\cdot\rank(R_1\otimes S_1+R_2\otimes S_2).
\end{eqnarray}  
The equation holds if and only if the rank is two, and one of the following two cases hold.
\begin{align}
    & (i) R_1 = \begin{bmatrix} 1 & 0 \\ 0 & 0 \end{bmatrix}, \quad 
      R_2 = \begin{bmatrix} 0 & 1 \\ 0 & 0 \end{bmatrix}, \quad
      \rank[S_1, S_2] = 1, \quad \rank[S_1^T, S_2^T] = 2; \label{eq:case1} \\
    & (ii) R_1 = \begin{bmatrix} 1 & 0 \\ 0 & 0 \end{bmatrix}, \quad 
      R_2 = \begin{bmatrix} 0 & 0 \\ 1 & 0 \end{bmatrix}, \quad
      \rank[S_1, S_2] = 2, \quad \rank[S_1^T, S_2^T]= 1. \label{eq:case2}
\end{align}\\
\end{lemma}
\begin{proof}
In \eqref{eq:SR=2,Ri=Si=2X2}, under the action of invertible matrices $P, Q \in \mathrm{M}_2$ such that ($R_i \mapsto P R_i Q $) and linear combinations ($R_i \mapsto a R_1 + b R_2 $ ), we may fix: $R_1 = \begin{bmatrix} 1 & 0 \\ 0 & 0 \end{bmatrix}, \quad \text{with } R_2 \text{ being a non-symmetric matrix}. $ In this case, the general form of $R_2$ is $R_2 = \begin{bmatrix} 0 & b \\ c & d \end{bmatrix}, \quad \text{where } b \neq c.$

We use Lemma \ref{le:X=A+yB} and \eqref{eq:SR=2,Ri=Si=2X2} to derive that the span of $R_1$ and $R_2$ has a singular matrix of order two. Up to equivalence and rename of $R_j,S_j$, we may regard the matrix as $R_1=\bma 1 & 0 \\ 0 & 0 \ema$, and obtain $R_2=\bma 0 & b \\ c & d \ema$. If $d\ne0$ then up to equivalence we can assume that $b=c=0$ and $d=1$. On the other hand if $d=0$ then up to equivalence we have three cases, namely (i) $b=1,c=0$, (ii) $b=0,c=1$, and (iii) $b=c=1$. 

Let $d\ne0$. We have $R_1=\left[\begin{matrix}1&0\\0&0\\\end{matrix}\right],R_2=\left[\begin{matrix}0&0\\0&1\\\end{matrix}\right]$. Then we have $R_1\otimes S_1^\G+R_2\otimes S_2^\G=\left[\begin{matrix}{S_1^\G}&O\\O&{S_2^\G}\\\end{matrix}\right]$ where $\left[\begin{matrix}{S_1^\G}&O\\O&{S_2^\G}\\\end{matrix}\right]$ is a $4\times 4$ matrix. Besides, $R_1\otimes S_1+R_2\otimes S_2=\left[\begin{matrix}{S_1}&O\\O&{S_2}\\\end{matrix}\right]$ where $\left[\begin{matrix}{S_1}&O\\O&{S_2}\\\end{matrix}\right]$ is a $4\times 4$ matrix. It is obvious that the two matrices have the same rank and fail to satisfy \eqref{eq:SR=2,Ri=Si=2X2}. It remains to prove the case (i)-(iii) below \eqref{eq:SR=2,Ri=Si=2X2}. We respectively discuss them as follows.

(i) Using \eqref{eq:SR=2,Ri=Si=2X2}, we have $R_1=\left[\begin{matrix}1&0\\0&0\\\end{matrix}\right],R_2=\left[\begin{matrix}0&1\\0&0\\\end{matrix}\right]$. Now we have $\rank(X^\G)\le2$ and $\rank(X)\le 2$ where $X^\G$ is $R_1\otimes S_1^T+R_2\otimes S_2^T$ and $X$ is $R_1\otimes S_1+R_2\otimes S_2$. Therefore, the inequality in \eqref{cj:1} is saturated if and only if $\rank(X^\G)=2, \rank(X)= 1$. Thus,
$R_1\otimes S_1^T+R_2\otimes S_2^T=\left[\begin{matrix}{S_1^\G}&{S_2^T}\\O&O\\\end{matrix}\right]$ where $\left[\begin{matrix}{S_1^T}&{S_2^T}\\O&O\\\end{matrix}\right]$ is a $4\times 4$ matrix, 
and
$R_1\otimes S_1+R_2\otimes S_2=\left[\begin{matrix}{S_1}&{S_2}\\O&O\\\end{matrix}\right]$ where $\left[\begin{matrix}{S_1}&{S_2}\\O&O\\\end{matrix}\right]$ is a $4\times 4$ matrix.

Besides, there exists a $2\times 2$ invertible matrix $W$ such that $WS_1=\left[\begin{matrix}a_{11}&a_{12}\\0&0\\\end{matrix}\right],WS_2=\left[\begin{matrix}b_{11}&b_{12}\\0&0\\\end{matrix}\right]$. Hence,
$\rank(R_1\otimes S_1^T+R_2\otimes S_2^T)=\rank(\left[\begin{matrix}{\left[\begin{matrix}{a_{11}}&0\\{a_{12}}&0\\\end{matrix}\right]}&{\left[\begin{matrix}{b_{11}}&0\\{b_{12}}&0\\\end{matrix}\right]}\\0&0\\\end{matrix}\right])=2$,
and
$\rank(R_1\otimes S_1+R_2\otimes S_2)=\rank(\left[\begin{matrix}{\left[\begin{matrix}{a_{11}}&{a_{12}}\\0&0\\\end{matrix}\right]}&{\left[\begin{matrix}{b_{11}}&{b_{12}}\\0&0\\\end{matrix}\right]}\\0&0\\\end{matrix}\right])=1$ in which $a_{11}b_{12}-a_{12}b_{11}\neq 0$.

(ii) In \eqref{eq:SR=2,Ri=Si=2X2}, we have $R_1=\left[\begin{matrix}1&0\\0&0\\\end{matrix}\right],R_2=\left[\begin{matrix}0&0\\1&0\\\end{matrix}\right]$. Similar to (i), the question can be transformed to a simple one as follows. 

$\rank(R_1\otimes S_1^T+R_2\otimes S_2^T)=\rank\left[\begin{matrix}{S_1^T}&O\\{S_2^T}&O\\\end{matrix}\right]=\rank\left[\begin{matrix}{\left[\begin{matrix}{a_{11}}&{a_{12}}\\0&0\\\end{matrix}\right]}&O\\{\left[\begin{matrix}{b_{11}}&{b_{12}}\\0&0\\\end{matrix}\right]}&O\\\end{matrix}\right]=2$,
and
$\rank(R_1\otimes S_1+R_2\otimes S_2)=\rank\left[\begin{matrix}{S_1}&O\\{S_2}&O\\\end{matrix}\right])=\rank(\left[\begin{matrix}{\left[\begin{matrix}{a_{11}}&0\\{a_{12}}&0\\\end{matrix}\right]}&O\\{\left[\begin{matrix}{b_{11}}&0\\{b_{12}}&0\\\end{matrix}\right]}&O\\\end{matrix}\right]=1$ in which $a_{11}b_{12}-a_{12}b_{11}\neq 0$.\\

(iii) In \eqref{eq:SR=2,Ri=Si=2X2}, we have
$R_1=\left[\begin{matrix}1&0\\0&0\\\end{matrix}\right],R_2=\left[\begin{matrix}0&1\\1&0\\\end{matrix}\right]$. One can show that $\rank X=\rank X^\G$, and \eqref{eq:SR=2,Ri=Si=2X2} does not hold. 

Conversely, if one of \eqref{eq:case1} and \eqref{eq:case2}  holds, then the equation necessarily hold.

So 
\eqref{eq:SR=2,Ri=Si=2X2} holds if and only if the rank is two, and one of \eqref{eq:case1} and \eqref{eq:case2} cases hold. \\

We conclude the proof so far.
Let $R_1,R_2\in\bbM_2$ be linearly independent, and $S_1,S_2\in\bbM_2$ be also linearly independent in  \eqref{eq:SR=2,Ri=Si=2X2}. Then the equation holds if and only if the rank is two, and one of the following two cases hold. 

(i) $R_1=\left[\begin{matrix}1&0\\0&0\\\end{matrix}\right],R_2=\left[\begin{matrix}0&1\\0&0\\\end{matrix}\right]$ and there exists an invertible matrix $W\in \bbM_{2}$ such that 

\begin{equation*}
W\left[\begin{matrix}S_1&S_2\\\end{matrix}\right]=\left[\begin{matrix}a_{11}&a_{12}&b_{11}&b_{12}\\0&0&0&0\\\end{matrix}\right], 
\end{equation*}

and $a_{11}b_{12}-a_{12}b_{11}\neq 0$;

(ii) $R_1=\left[\begin{matrix}1&0\\0&0\\\end{matrix}\right],R_2=\left[\begin{matrix}0&0\\1&0\\\end{matrix}\right]$, and there exists an invertible matrix $Q\in \bbM_{2}$ such that 

\begin{equation*}
\left[\begin{matrix}S_1\\S_2\\\end{matrix}\right]Q=\left[\begin{matrix}a_{11}&0\\a_{21}&0\\b_{11}&0\\b_{21}&0\\\end{matrix}\right], 
\end{equation*}

and $a_{11}b_{21}-a_{21}b_{11}\neq 0$.\\
\end{proof}

Next we extend the above lemma to order-four matrices of Schmidt rank three and four. 

\begin{lemma}
\label{le:SR=3+4,Ri=Si=2X2}  
(i) If the Schmidt rank of order-four matrix $X$ has rank three, then the inequality in Theorem \ref{cj:1} is strict.   

(ii) If the Schmidt rank of order-four matrix $X$ has rank four, then the inequality in Theorem \ref{cj:1} is saturated if and only if $X$ has rank one.   
\end{lemma} 
\begin{proof}
(i) We follow the argument in the first paragraph of proof of Lemma \ref{le:SR=2,Ri=Si=2X2}. We obtain that $X$ is equivalent to one of the following two cases, 
\begin{eqnarray}
&&
X_1=\bma 1&0\\0&0\ema \otimes S_1+ 
\bma 0&0\\0&1\ema \otimes S_2+
\bma 0& 1\\1&0\ema \otimes S_3,
\\&&
X_2=\bma 1&0\\0&0\ema \otimes S_1+ 
\bma 0&1\\0&0\ema \otimes S_2+
\bma 0& 0\\1&0\ema \otimes S_3.
\end{eqnarray}
To make the inequality in Theorem \ref{cj:1} saturated, we obtain that $\rank X=1$. One can show that this is impossible by the expression of nonzero $S_j$'s in $X_1$ and $X_2$.

(ii) The proof is similar to that of (i).   
\end{proof}

\subsection{$m_2\times n_2$ matrices $S_i$}
\label{sec:2xd}

In this subsection, we study \eqref{eq:SR=2,Ri=Si=2X2} when $R_i$ are $2\times2$ matrices, and $S_i$ are $m_2\times n_2$ matrices.

\begin{lemma}
\label{le:SR=2,Ri=2X2}
Let $R_1,R_2\in\bbM_2$ be linearly independent, and $S_1,S_2\in\bbM_{m_2,n_2}$ be also linearly independent in the equation \eqref{eq:SR=2,Ri=Si=2X2}. It holds if and only if  one of the following two cases hold. 

(i) $\rank \left[\begin{matrix}S_1\\S_2\\\end{matrix}\right]=d$ and $\rank\left[\begin{matrix}S_1&S_2\\\end{matrix}\right]=2d$ where $d=1,2,...2m$ and $2m$ is the largest even number not exceeding $m_2$. There is an invertible matrix $Q\in \bbM_{n_2}$ such that 

\bea
\left[\begin{matrix}S_1\\S_2\\\end{matrix}\right]Q=\left[\begin{matrix}\vec{u_1}&\vec{u_2}&\ldots&\vec{u_d}&0&\ldots&0\\\vec{v_1}&\vec{v_2}&\ldots&\vec{v_d}&0&\ldots&0\\\end{matrix}\right], 
\eea

where $\vec{u_1},\vec{u_2},...\vec{u_d},\vec{v_1},\vec{v_2},...\vec{v_d}\in \bbC^{m_2}$ are linearly independent column vectors.

(ii) $\rank \left[\begin{matrix}S_1&S_2\\\end{matrix}\right]=d$ and $\rank \left[\begin{matrix}S_1\\S_2\\\end{matrix}\right]=2d$, where $d=1,2,...2m$ and $2m$ is the largest even number not exceeding $n_2$. There is an invertible matrix $W\in \bbM_{m_2}$ such that

\bea W\left[\begin{matrix}S_1&S_2\\\end{matrix}\right]=\left[\begin{matrix}{\vec{u_1}}^T&{\vec{v_1}}^T\\{\vec{u_2}}^T&{\vec{v_2}}^T\\\vdots&\vdots\\{\vec{u_d}}^T&{\vec{v_d}}^T\\0&0\\\vdots&\vdots\\0&0\\\end{matrix}\right], 
\eea

where $\vec{u_1},\vec{u_2},...\vec{u_d},\vec{v_1},\vec{v_2},...\vec{v_d}\in \bbC^{n_2}$ are linearly independent column vectors.
\end{lemma}
\begin{proof}
Let $M=\sum^2_{j=1}R_j\otimes S_j$ have Schmidt rank two. To saturate inequality \eqref{eq:cj2}, similar to the argument of the first paragraph of proof of Lemma \ref{le:SR=2,Ri=Si=2X2}, we can assume that $R_1=\bma 1 & 0 \\ 0 & 0 \ema$ and $R_2=\bma 0 & b \\ c & 0 \ema$. Up to equivalence we have three cases, namely (i) $b=0,c=1$, (ii) $b=1,c=0$, and (iii) $b=c=1$. Using the matrix transposition operation on both sides of \eqref{eq:cj2}, we see that case (i) and (ii) are equivalent. So we only study (ii) and (iii). 

In case (i), the matrices in the two sides of \eqref{eq:cj2} become
\begin{eqnarray}
\bma
S_1^T \\ S_2^T
\ema,
\quad
\bma
S_1 \\ S_2
\ema.
\end{eqnarray}
The saturation of inequality \eqref{eq:cj2} means
\begin{eqnarray}
\label{S1^T S2^T S_1 S_2}
\rank
\bma
S_1^T \\ S_2^T
\ema
=2\rank \bma
S_1 \\ S_2
\ema.
\end{eqnarray}
This expression matches the form of the earlier Theorem \ref{le:ANYm2,K>2}, and the equation holds if $\rank \left[\begin{matrix}S_1\\S_2\\\end{matrix}\right]=d$ and $\rank\left[\begin{matrix}S_1&S_2\\\end{matrix}\right]=2d$ where $d=1,2,...2m$ and $2m$ is the largest even number not exceeding $m_2$. Respectively, there exists an invertible matrix $Q\in \bbM_{n_2}$ such that 

\bea
\left[\begin{matrix}S_1\\S_2\\\end{matrix}\right]Q=\left[\begin{matrix}\vec{u_1}&\vec{u_2}&\ldots&\vec{u_d}&0&\ldots&0\\\vec{v_1}&\vec{v_2}&\ldots&\vec{v_d}&0&\ldots&0\\\end{matrix}\right], 
\eea

where $\vec{u_1},\vec{u_2},...\vec{u_d},\vec{v_1},\vec{v_2},...\vec{v_d}\in \bbC^{m_2}$ are linearly independent column vectors such that $\left[\begin{matrix}{\vec{u_1}}^T\\{\vec{u_2}}^T\\\vdots\\{\vec{u_d}}^T\\{\vec{v_1}}^T\\\vdots\\{\vec{v_d}}^T\\\end{matrix}\right]$ is a matrix of full rank. These column vectors $\vec{u_1},\vec{u_2},...\vec{u_d},\vec{v_1}\in \bbC^{m_2}$ are precisely the columns of the new matrix obtained after performing column transformations of $Q$ on the original matrix $S_1$ and $\vec{v_1},\vec{v_2},...\vec{v_d}\in \bbC^{m_2}$ are precisely the columns of the new matrix obtained after performing column transformations of $Q$ on the original matrix $S_2$.

On the other hand, one can obtain the equality condition using Theorem \ref{le:ANYm2,K>2}. A similar discussion applies when $S_i$ is a row or column vector due to the switch of $R_i$ and $S_i$, in which we obtain case (ii).

In case (iii), the matrices in the two sides of \eqref{eq:cj2} become
\begin{eqnarray}
\bma
\begin{matrix}{S_1}^T&{S_2}^T\\{S_2}^T&O\\\end{matrix}
\ema,
\quad
\bma
\begin{matrix}{S_1}&{S_2}\\{S_2}&O\\\end{matrix}
\ema.
\end{eqnarray}
The saturation of inequality \eqref{eq:cj2} means
\begin{eqnarray}
\rank
\bma
\begin{matrix}{S_1}^T&{S_2}^T\\{S_2}^T&O\\\end{matrix}
\ema
=2\rank \bma
\begin{matrix}{S_1}&{S_2}\\{S_2}&O\\\end{matrix}
\ema.
\end{eqnarray}
As the matrices are transposed versions of each other, the equation does not hold.
\end{proof}

Next we extend the above lemma to order-four matrices of Schmidt rank three and four. 

\begin{lemma}
\label{le:sr=4}
If the Schmidt rank of order-four matrix $X$ is three, then the inequality in Theorem \ref{cj:1} is strict.   
\end{lemma} 
\begin{proof}
We follow the argument in the first paragraph of proof of Lemma \ref{le:SR=2,Ri=Si=2X2}. We obtain that $X$ is equivalent to one of the following two cases. To make the inequality in Theorem \ref{cj:1} saturated, firstly, we consider $X_1=\bma 1&0\\0&0\ema \otimes S_1+ 
\bma 0&0\\0&1\ema \otimes S_2+
\bma 0& 1\\1&0\ema \otimes S_3$. But $X_1$, being a symmetric matrix, clearly cannot satisfy the equation \eqref{cj:1}. Secondly, we consider $X_2=\bma 1&0\\0&0\ema \otimes S_1+ 
\bma 0&1\\0&0\ema \otimes S_2+
\bma 0& 0\\1&0\ema \otimes S_3=\bma S_1&S_2\\S_3&O\ema\in\mathbb{M}_2\otimes \mathbb{M}_{m_2\times n_2}$. Correspondingly, $X_2^{\G}=\bma S_1&S_3\\S_2&O\ema$. Then we discuss whether the following equation holds and the required conditions.
\begin{align*}
\label{x2}
\rank (X_2^\G)=3\cdot \rank (X_2)
\end{align*}
We assume $\rank S_1=f$, thus there exists an invertible matrix $W$ of order $m_2$ such that $WS_1=\diag (I_f,O)$ where $f\le \min \{m_2,n_2\}$. Correspondingly, suppose $WS_2=S_2',WS_3=S_3'$ in which $\rank S_2=\rank S_2',\rank S_3=\rank S_3'$. Furthermore, we construct an invertible matrix $\diag (W,I_{m_2})$ of order $2m_2$ such that 
\bea
\left[\begin{matrix}W&O\\O&I_{m_2}\\\end{matrix}\right]\left[\begin{matrix}S_1&S_2\\S_3&O\\\end{matrix}\right]=\left[\begin{matrix}WS_1&WS_2\\S_3&O\\\end{matrix}\right]
=\left[\begin{matrix}\left[\begin{matrix}I_f&O\\O&O\\\end{matrix}\right]&S_2'\\S_3&O\\\end{matrix}\right],
\eea
\bea
\left[\begin{matrix}W&O\\O&I_{m_2}\\\end{matrix}\right]\left[\begin{matrix}S_1&S_3\\S_2&O\\\end{matrix}\right]=\left[\begin{matrix}WS_1&WS_3\\S_2&O\\\end{matrix}\right]=\left[\begin{matrix}\left[\begin{matrix}I_f&O\\O&O\\\end{matrix}\right]&S_3'\\S_2&O\\\end{matrix}\right].
\eea
Therefore if equation holds, we have 
\begin{align*}
    f+\rank S_2+\rank S_3
    &\geq \rank (X_2^\G)
    =3\cdot \rank (X_2)\\
    &=2\cdot \rank (X_2)+\rank (X_2)
    \geq2(f+\rank S_2)+f+\rank S_3. 
\end{align*}
So we obtain $0\geq2f+\rank S_2$, which means $f=\rank S_2=0$. Similarly, we obtain $0\geq2f+\rank S_3$, which means $f=\rank S_3=0$. Thus $X_2=O$. This leads to a contradiction.

In conclusion, if the Schmidt rank of order-four matrix $X$ is three, then the inequality in Theorem \ref{cj:1} is strict.  
\end{proof}

Next, we investigate the saturation of inequality in Theorem \ref{cj:1} when $X$ has Schmidt rank four.

\begin{lemma}
\label{le:m1=n1=2,sr=4}
If the matrix $X=\left[\begin{matrix}X_1&X_2\\X_3&X_4\\\end{matrix}\right]\in\bbM_2\otimes\bbM_{m_2,n_2}$ has Schmidt rank four, then the inequality in Theorem \ref{cj:1} is saturated if and only if  $r=\rank X=\rank X_j$ for $j=1,2,3,4$ up to local equivalence. Further $X$ is locally equivalent to
 \begin{eqnarray}
    \label{X_to_4r}
\left[\begin{matrix}\left[\begin{matrix}I_r&O\\O&O\\\end{matrix}\right]&\left[\begin{matrix}O&I_r\\O&O\\\end{matrix}\right]\\\left[\begin{matrix}O&O\\I_r&O\\\end{matrix}\right]&\left[\begin{matrix}O&O\\O&I_r\\\end{matrix}\right]\\\end{matrix}\right].
\end{eqnarray}
Here each block is an order-$2r$ matrix, by removing additional zero columns and rows of the corresponding $X_j$.
\end{lemma}

\begin{proof}
It is equivalent to $\rank X^\G=4\rank X$ where $X=\left[\begin{matrix}X_1&X_2\\X_3&X_4\\\end{matrix}\right],X^\G=\left[\begin{matrix}X_1&X_3\\X_2&X_4\\\end{matrix}\right]$. 
Using the known facts on matrix rank, we have
\begin{eqnarray}
\rank X^\G \le 
\rank [X_1, X_3]
+
\rank [X_2, X_4]
\le
\sum^4_{j=1} \rank X_j
\le 
\sum^4_{j=1} \rank X
=4 \rank X.
\notag\\
\end{eqnarray}
Because the equality in \eqref{eq:cj2} holds, the above inequalities are all saturated. We have
\begin{eqnarray}
\label{eq:rankXj=rankX}
&&
{1\over4}\rank X^\G=\rank X_j=\rank X, \quad j=1,2,3,4,  
\\=&& 
{1\over2}\rank[X_1,X_3]
=
{1\over2}\rank[X_2,X_4]
\\=&& 
{1\over2}\rank\bma X_1 \\ X_2\ema
= 
{1\over2}\rank\bma X_3 \\ X_4\ema.
\end{eqnarray}
Suppose $\rank X_1=r$. The above equations imply that 
\begin{eqnarray}
    \label{eq:RX_decomposition1}
    R(X) &=& R\begin{bmatrix}X_1\\X_3\end{bmatrix}
    = R\begin{bmatrix}X_2\\X_4\end{bmatrix}, \\    \label{eq:RX_decomposition2}
    R(X_1^T) &=& R(X_3^T), \quad
    R(X_2^T) = R(X_4^T),
\end{eqnarray}
The implication is due to Lemma \ref{le:rankA=rankA1}. We see that $\rank X_1=\rank X$ holds if and only if $R(X)=R\begin{bmatrix}X_1\\X_3\\\end{bmatrix}\ and\ R(X_3^T)\subseteq R(X_1^T)$. Similarly, $\rank X_3=\rank X$ holds if and only if $R(X)=R\begin{bmatrix}X_1\\X_3\\\end{bmatrix}\ and\ R(X_1^T)\subseteq R(X_3^T)$. For $\rank X_2=\rank X,\rank X_4=\rank X$, the equivalent conditions exist. So $\rank X_j=\rank X, \quad j=1,2,3,4$ holds if and only if $R(X)=R\begin{bmatrix}X_1\\X_3\\\end{bmatrix}=R\begin{bmatrix}X_2\\X_4\\\end{bmatrix}$ and $R(X_1^T)=R(X_3^T),R(X_2^T)=R(X_4^T)$.

Using Lemma \ref{le:rankA=rankB+rankC}, we can also obtain the equation as follows.
\begin{eqnarray}
R(X^\G) =
R\bma X_1 \\ X_2 \ema+R\bma X_3 \\ X_4 \ema.
\end{eqnarray}
Moreover, the dimension of the vector space satisfies the following conditions.
\begin{eqnarray}
\label{dim4r&2r}
\dim R(X^\G) = 4r,
\\ \dim R\bma X_1 \\ X_2 \ema = R\bma X_3 \\ X_4 \ema=2r. 
\end{eqnarray}

Using SVD on $X_1$ and Lemma \ref{le:rankA=rankA1}, we may assume that
\begin{eqnarray}
    \label{X_block1}
    &\rank X=\rank \left[\begin{matrix}X_1&X_2\\X_3&X_4\\\end{matrix}\right]=\rank\left[\begin{matrix}\left[\begin{matrix}I_r&O\\O&O\\\end{matrix}\right]&\left[\begin{matrix}X_{21}&X_{22}\\O&O\\\end{matrix}\right]\\\left[\begin{matrix}X_{31}&O\\X_{33}&O\\\end{matrix}\right]&\left[\begin{matrix}X_{41}&X_{42}\\X_{43}&X_{44}\\\end{matrix}\right]\\\end{matrix}\right],\\
    \label{X_block2}
    &\rank X^\G=\rank \left[\begin{matrix}X_1&X_3\\X_2&X_4\\\end{matrix}\right]=\rank\left[\begin{matrix}\left[\begin{matrix}I_r&O\\O&O\\\end{matrix}\right]&\left[\begin{matrix}X_{31}&O\\X_{33}&O\\\end{matrix}\right]\\\left[\begin{matrix}X_{21}&X_{22}\\O&O\\\end{matrix}\right]&\left[\begin{matrix}X_{41}&X_{42}\\X_{43}&X_{44}\\\end{matrix}\right]\\\end{matrix}\right].
\end{eqnarray}
Next, we reduce the size of matrix $X$. In $\rank\left[\begin{matrix}\left[\begin{matrix}I_r&O\\O&O\\\end{matrix}\right]&\left[\begin{matrix}X_{21}&X_{22}\\O&O\\\end{matrix}\right]\\\left[\begin{matrix}X_{31}&O\\X_{33}&O\\\end{matrix}\right]&\left[\begin{matrix}X_{41}&X_{42}\\X_{43}&X_{44}\\\end{matrix}\right]\\\end{matrix}\right]$, since columns $n_2-r$ to $n_2$ are all zero, the Kronecker product’s properties allow local operations on $X_2$ and $X_4$, such that $\left[\begin{matrix}\left[\begin{matrix}X_{21}&X_{22}\\O&O\\\end{matrix}\right]\\\left[\begin{matrix}X_{41}&X_{42}\\X_{43}&X_{44}\\\end{matrix}\right]\\\end{matrix}\right]$ can be reduced to $\left[\begin{matrix}\left[\begin{matrix}X_{21}&X_{22}'&O\\O&O&O\\\end{matrix}\right]\\\left[\begin{matrix}X_{41}&X_{42}'&O\\X_{43}&X_{44}'&O\\\end{matrix}\right]\\\end{matrix}\right]$ where $X_{21},X_{22}$ are both order-$r$ matrix and $X_{42}',X_{44}'$ both have $r$ columns. Similarly, we can perform operations on rows $2m_2-r$ through $2m_2$. Therefore, there exist order-$m_2$ and order-$n_2$ invertible matrices $\left[\begin{matrix}I_r&O\\O&B_1\\\end{matrix}\right],\left[\begin{matrix}I_r&O\\O&B_2\\\end{matrix}\right]$ such that
\begin{align*}
&(I_2\otimes \left[\begin{matrix}I_r&O\\O&B_1\\\end{matrix}\right])\bigg( \sum^4_{i=1} R_i \ox S_i \bigg)(I_2\otimes \left[\begin{matrix}I_r&O\\O&B_2\\\end{matrix}\right])\\
= &\sum^4_{i=1} (I_2R_iI_2) \ox \left[\begin{matrix}I_r&O\\O&B_1\\\end{matrix}\right]S_i\left[\begin{matrix}I_r&O\\O&B_2\\\end{matrix}\right] 
=\left[\begin{matrix}\left[\begin{matrix}I_r&O&O\\O&O&O\\O&O&O\\\end{matrix}\right]&\left[\begin{matrix}X_{21}&X_{22}''&O\\O&O&O\\O&O&O\\\end{matrix}\right]\\\left[\begin{matrix}X_{31}&O&O\\X_{33}''&O&O\\O&O&O\\\end{matrix}\right]&\left[\begin{matrix}X_{41}&X_{42}''&O\\X_{43}''&X_{44}''&O\\O&O&O\\\end{matrix}\right]\\\end{matrix}\right]. 
\end{align*}
Thus, we have 
\begin{eqnarray}
\label{Xtosquare}
\rank X=\rank \left[\begin{matrix}\left[\begin{matrix}I_r&O\\O&O\\\end{matrix}\right]&\left[\begin{matrix}X_{21}&X_{22}\\O&O\\\end{matrix}\right]\\\left[\begin{matrix}X_{31}&O\\X_{33}&O\\\end{matrix}\right]&\left[\begin{matrix}X_{41}&X_{42}\\X_{43}&X_{44}\\\end{matrix}\right]\\\end{matrix}\right], 
\end{eqnarray}
where $X_{ij}$ is order-$r$ matrix and $X_{22}''$ has already been expressed in the form of $X_{22}$, similar for $X_{33}'',X_{42}'',X_{43}'',X_{44}''$. 

Using \eqref{dim4r&2r}, we obtain $\dim R\left[\begin{matrix}\left[\begin{matrix}I_r&O\\O&O\\\end{matrix}\right]\\\left[\begin{matrix}X_{21}&X_{22}\\O&O\\\end{matrix}\right]\\\end{matrix}\right]=2r$, which means $\rank X_{22}=r$. So $\rank X=\rank \left[\begin{matrix}\left[\begin{matrix}I_r&O\\O&O\\\end{matrix}\right]&\left[\begin{matrix}O&X_{22}\\O&O\\\end{matrix}\right]\\\left[\begin{matrix}X_{31}&O\\X_{33}&O\\\end{matrix}\right]&\left[\begin{matrix}X_{41}&X_{42}\\X_{43}&X_{44}\\\end{matrix}\right]\\\end{matrix}\right]=\rank \left[\begin{matrix}\left[\begin{matrix}I_r&O\\O&O\\\end{matrix}\right]&\left[\begin{matrix}O&I_r\\O&O\\\end{matrix}\right]\\\left[\begin{matrix}X_{31}&O\\X_{33}&O\\\end{matrix}\right]&\left[\begin{matrix}O&X_{31}\\O&X_{33}\\\end{matrix}\right]\\\end{matrix}\right]$ where we have already used $\rank X=r$. Next, using \eqref{eq:rankXj=rankX}, $\rank \left[\begin{matrix}\left[\begin{matrix}I_r&O\\O&O\\\end{matrix}\right]&\left[\begin{matrix}X_{31}&O\\X_{33}&O\\\end{matrix}\right]\\\end{matrix}\right]=2r$ which means $\rank X_{33}=r$. Therefore, matrix $X$ can be further reduced to the form as follows. 
\begin{eqnarray}
    \label{X_to_4r}
    \rank X=\rank\left[\begin{matrix}\left[\begin{matrix}I_r&O\\O&O\\\end{matrix}\right]&\left[\begin{matrix}O&I_r\\O&O\\\end{matrix}\right]\\\left[\begin{matrix}O&O\\I_r&O\\\end{matrix}\right]&\left[\begin{matrix}O&O\\O&I_r\\\end{matrix}\right]\\\end{matrix}\right].
\end{eqnarray}
\end{proof}

\subsection{A more general case on the equality condition}

\begin{lemma}
\label{k^2Schmidt_rank}
Suppose the matrix $X=\left[\begin{matrix}X_{11}&X_{12}&\ldots&X_{1k}\\X_{21}&X_{22}&\ldots&X_{2k}\\\vdots&\vdots&\vdots&\vdots\\X_{k1}&X_{k2}&\ldots&X_{kk}\\\end{matrix}\right]
\in\bbM_k\otimes\bbM_{m_2,n_2}$ has Schmidt rank $k^2$. Then the inequality in Theorem \ref{cj:1} is saturated if and only if $r=\rank X=\rank X_{ij}$ for $i,j=1,2,...k$ up to local equivalence. Further $X$ is locally equivalent to
\begin{eqnarray}
\label{XEij} \left[\begin{matrix}E_{11}&E_{12}&\ldots&E_{1k}\\E_{21}&E_{22}&\ldots&E_{2k}\\\vdots&\vdots&\vdots&\vdots\\E_{k1}&E_{k2}&\ldots&E_{kk}\\\end{matrix}\right].
\end{eqnarray}
Here each $E_{ij}$ is an $kr\times kr$ block matrix partitioned into $k^2$ blocks where the block at the $(i, j)$ position is the order-$r$ identity matrix $I_r$, and all other blocks are zero, by removing additional zero columns and rows of the corresponding $X_{ij}$.
\end{lemma}
\begin{proof}
The 'if' part can be verified straightforwardly. We prove the 'only if' part. We adopt the following notation. Suppose $E_{ij}$ is an $kr\times kr$ block matrix partitioned into $k^2$ blocks where the block at the $(i, j)$ position is the order-$r$ identity matrix $I_r$. For example, $E_{11}=\left[\begin{matrix}I_r&O&\ldots&O\\O&O&\ldots&O\\\vdots&\vdots&\vdots&\vdots\\O&O&\ldots&O\\\end{matrix}\right]$. Next, we apply mathematical induction on $k$, the order of the square matrix $R_i$.

(i) When $k=2$, the statement holds by Theorem \ref{le:m1=n1=2,sr=4}. 

(ii) Let $k>2$. Suppose the statement is true for $R_i$ of size $(k-1) \times (k-1)$. Therefore similar to Theorem \ref{le:m1=n1=2,sr=4}, through elementary operations $X$ can be reduced to the form $\left[\begin{matrix}E_{11}'&\ldots&E_{1,k-1}'&A_1'\\\vdots&\ldots&\vdots&\vdots\\E_{k-1,1}'&\ldots&E_{k-1,k-1}'&A_{k-1}'\\B_1'&\ldots&B_{k-1}'&W'\\\end{matrix}\right]$ where $A_i', B_i', i=1,2,...k-1$ and $W'$ are all $m_2\times n_2$ matrices. Here we adopt the notation $E_{ij}'$ because it has not yet been reduced to a $kr$-order square matrix. Since the up-left $(k-1)^2$ block sub-matrices of the above matrix $\left[\begin{matrix}E_{11}'&\ldots&E_{1,k-1}'&A_1'\\\vdots&\ldots&\vdots&\vdots\\E_{k-1,1}'&\ldots&E_{k-1,k-1}'&A_{k-1}'\\B_1'&\ldots&B_{k-1}'&W'\\\end{matrix}\right]$ have their last $(m_2-kr)$ rows and $(n_2-kr)$ columns as zero matrices, similar to the previous form \eqref{Xtosquare}, we may regard the entire large matrix as being composed of $k^2$ square sub-matrices of order-$kr$. Then we obtain 
\begin{eqnarray}
\rank X=\rank \left[\begin{matrix}E_{11}&\ldots&E_{1,k-1}&A_1\\\vdots&\ldots&\vdots&\vdots\\E_{k-1,1}&\ldots&E_{k-1,k-1}&A_{k-1}\\B_1&\ldots&B_{k-1}&W\\\end{matrix}\right], 
\end{eqnarray}
where each $A_i,i=1,2,...k$ is an order-$kr$ matrix, whose $i$-th row is $\left[\begin{matrix}A_{i1}&A_{i2}&\ldots&A_{ik}\\\end{matrix}\right]$ and the other block matrices in it are order-$r$ zero matrices. Also, each $B_i,i=1,2,...k$ is an order-$kr$ matrix, whose $i$-th column is $\left[\begin{matrix}B_{1i}\\B_{2i}\\\vdots\\B_{ki}\\\end{matrix}\right]$ and the other block matrices in it are order-$r$ zero matrices. Then we consider $\rank X^\G$. We have that $\rank X^\G=k^2\rank X$ is equivalent to the case where all inequalities below become equalities. 
\begin{eqnarray}
\label{k_=}
&\rank X^\G=\rank \left[\begin{matrix}E_{11}&\ldots&E_{k-1,1}&B_1\\\vdots&\ldots&\vdots&\vdots\\E_{1,k-1}&\ldots&E_{k-1,k-1}&B_{k-1}\\A_1&\ldots&A_{k-1}&W\\\end{matrix}\right] \\
&\le \sum_{j=1}^{k-1}\rank \left[\begin{matrix}E_{1j}&\ldots&E_{k-1,j}&B_j\\\end{matrix}\right]+\rank \left[\begin{matrix}A_1&\ldots&A_{k-1}&W\\\end{matrix}\right]\\
&\le \sum_{j=1}^{k^2}\rank X_j\le \sum_{j=1}^{k^2}\rank X=k^2\rank X. 
\end{eqnarray}
Thus we know that $\rank \left[\begin{matrix}E_{1j}&\ldots&E_{k-1,j}&B_j\\\end{matrix}\right]=kr,j=1,2,...k-1$. When $j=1$, we have 
\begin{eqnarray}
\label{row_j=1}
\rank \left[\begin{matrix}\left[\begin{matrix}I_r&O&\ldots&O\\\vdots&\vdots&\ldots&O\\O&O&\ldots&\vdots\\O&O&\ldots&O\\\end{matrix}\right]&\ldots&\left[\begin{matrix}O&O&\ldots&O\\\vdots&\vdots&\ldots&O\\I_r&O&\ldots&\vdots\\O&O&\ldots&O\\\end{matrix}\right]&\left[\begin{matrix}B_{11}&O&\ldots&O\\\vdots&\vdots&\ldots&O\\B_{k-1,1}&O&\ldots&\vdots\\B_{k1}&O&\ldots&O\\\end{matrix}\right]\\\end{matrix}\right]=kr. 
\end{eqnarray}
So $\left[\begin{matrix}B_{11}&O&\ldots&O\\\vdots&\vdots&\ldots&O\\B_{k-1,1}&O&\ldots&\vdots\\B_{k1}&O&\ldots&O\\\end{matrix}\right]$ with rank-$r$ in \eqref{row_j=1} can be reduced to the form of  $E_{k1}$. Moreover, due to the presence of the zero matrices, the other matrices in \eqref{row_j=1} retain their original forms. Similarly, $B_2,...B_{k-1}$ with rank-$r$ can be reduced to the form of $E_{k2},...E_{k,k-1}$. Besides, the column space of $\left[\begin{matrix}E_{11}&\ldots&E_{k-1,1}&B_1\\\vdots&\ldots&\vdots&\vdots\\E_{1,k-1}&\ldots&E_{k-1,k-1}&B_{k-1}\\A_1&\ldots&A_{k-1}&W\\\end{matrix}\right]$ also satisfies conditions similar to those under which the above inequality \eqref{k_=} becomes an equality, which means $A_1,A_2,... A_{k-1}$ can be reduced to the form of $E_{1k},E_{2k},...E_{k-1,k}$. In the above process, matrix $w$ also undergoes corresponding elementary transformations and becomes $W''$, which do not alter its rank. 

Now let us revisit $\rank X$ that satisfies
\begin{eqnarray}
\rank X=\rank \left[\begin{matrix}E_{11}&\ldots&E_{1,k-1}&E_{1k}\\\vdots&\ddots&\vdots&\vdots\\E_{k-1,1}&\ldots&E_{k-1,k-1}&E_{k-1,k}\\E_{k1}&\ldots&E_{k,k-1}&W''\\\end{matrix}\right]\\
=\rank \left[\begin{matrix}\left[\begin{matrix}I_r&\ldots&O&O\\O&\ldots&O&O\\\vdots&\vdots&\vdots&\vdots\\O&\ldots&O&O\\\end{matrix}\right]&\ldots&\left[\begin{matrix}O&\ldots&I_r&O\\O&\ldots&O&O\\\vdots&\vdots&\vdots&\vdots\\O&\ldots&O&O\\\end{matrix}\right]&\left[\begin{matrix}O&\ldots&O&I_r\\O&\ldots&O&O\\\vdots&\vdots&\vdots&\vdots\\O&\ldots&O&O\\\end{matrix}\right]\\\left[\begin{matrix}O&\ldots&O&O\\I_r&\ldots&O&O\\\vdots&\vdots&\vdots&\vdots\\O&\ldots&O&O\\\end{matrix}\right]&\ldots&\left[\begin{matrix}O&\ldots&O&O\\O&\ldots&I_r&O\\\vdots&\vdots&\vdots&\vdots\\O&\ldots&O&O\\\end{matrix}\right]&\left[\begin{matrix}O&\ldots&O&O\\O&\ldots&O&I_r\\\vdots&\vdots&\vdots&\vdots\\O&\ldots&O&O\\\end{matrix}\right]\\\vdots&\ddots&\vdots&\vdots\\\left[\begin{matrix}O&\ldots&O&O\\O&\ldots&O&O\\\vdots&\vdots&\vdots&\vdots\\I_r&\ldots&O&O\\\end{matrix}\right]&\ldots&\left[\begin{matrix}O&\ldots&O&O\\O&\ldots&O&O\\\vdots&\vdots&\vdots&\vdots\\O&\ldots&I_r&O\\\end{matrix}\right]&W''\\\end{matrix}\right].
\end{eqnarray}
Additionally, we know $\rank X=r$, which means now order-$kr$ $W''$ has already become $E_{kk}=\left[\begin{matrix}O&\ldots&O&O\\O&\ldots&O&O\\\vdots&\vdots&\vdots&\vdots\\O&\ldots&O&I_r\\\end{matrix}\right]$. Therefore we obtain that when $X=\sum^{k^2}_{i=1} R_i \ox S_i$ has Schmidt rank $k^2$ where $R_i$ is  order-$k$and $S_i$ is a $m_2\times n_2$ matrix, the inequality in Theorem \ref{cj:1} is saturated if and only if $X$ has rank $r$ and $X$ is locally equivalent the simplest form as follows. 
\begin{eqnarray}
&\rank X^\G
=\rank \left[\begin{matrix}E_{11}&E_{21}&\ldots&E_{k1}\\E_{12}&E_{22}&\ldots&E_{k2}\\\vdots&\vdots&\vdots&\vdots\\E_{1k}&E_{2k}&\ldots&E_{kk}\\\end{matrix}\right]=k^2r\\
&=k^2\rank \left[\begin{matrix}E_{11}&E_{12}&\ldots&E_{1k}\\E_{21}&E_{22}&\ldots&E_{2k}\\\vdots&\vdots&\vdots&\vdots\\E_{k1}&E_{k2}&\ldots&E_{kk}\\\end{matrix}\right]=k^2\rank X.
\end{eqnarray}
\end{proof}

\begin{lemma}
\label{m1n1Schmidt_rank}
Suppose the matrix $X=\left[\begin{matrix}X_{11}&X_{12}&\ldots&X_{1n_1}\\X_{21}&X_{22}&\ldots&X_{2n_1}\\\vdots&\vdots&\vdots&\vdots\\X_{m_11}&X_{m_12}&\ldots&X_{m_1 n_1}\\\end{matrix}\right]\in\bbM_{m_1,n_1}\otimes\bbM_{m_2,n_2}$ has Schmidt rank $m_1 n_1\le m_2 n_2$. Then the inequality in Theorem \ref{cj:1} is saturated if and only if $r=\rank X=\rank X_{ij}$ for $i=1,2,...,m_1;j=1,2,... ,n_1$ up to local equivalence. Further $X$ is locally equivalent to
\begin{eqnarray}
\label{XEij} \left[\begin{matrix}E_{11}&E_{12}&\ldots&E_{1n_1}\\E_{21}&E_{22}&\ldots&E_{2n_1}\\\vdots&\vdots&\vdots&\vdots\\E_{m_11}&E_{m_12}&\ldots&E_{m_1 n_1}\\\end{matrix}\right].
\end{eqnarray}
Here each $E_{ij}$ is an $kr\times kr$ block matrix partitioned into $k^2$ blocks where $k=\max \{m_1,n_1\}$  , the block at the $(i, j)$ position is the order-$r$ identity matrix $I_r$, and all other blocks are zero, by removing additional zero columns and rows of the corresponding $X_{ij}$.
\end{lemma}
\begin{proof}
We have that $\rank X^\G=m_1 n_1\cdot\rank X$ is equivalent to the case where all inequalities below become equalities. 
\begin{eqnarray}
\label{m_1_n_1_k^2}
\rank X^\G\le \sum_{j=1}^{m_1 n_1}\rank X_j\le \sum_{j=1}^{m_1 n_1}\rank X=m_1 n_1\cdot \rank X.
\end{eqnarray}
Thus $X$ still can be reduced to the form as is demonstrated before. Then we obtain
\begin{eqnarray}
&\rank X^\G
=\rank \left[\begin{matrix}E_{11}&E_{21}&\ldots&E_{m_11}\\E_{12}&E_{22}&\ldots&E_{m_12}\\\vdots&\vdots&\vdots&\vdots\\E_{1n_1}&E_{2n_1}&\ldots&E_{m_1 n_1}\\\end{matrix}\right]=m_1 n_1r\\
&=m_1 n_1\rank \left[\begin{matrix}E_{11}&E_{12}&\ldots&E_{1n_1}\\E_{21}&E_{22}&\ldots&E_{2n_1}\\\vdots&\vdots&\vdots&\vdots\\E_{m_11}&E_{m_12}&\ldots&E_{m_1 n_1}\\\end{matrix}\right]=m_1 n_1\rank X.
\end{eqnarray}
\end{proof}

Next, let us consider an example. Suppose the matrix $X=\left[\begin{matrix}X_{11}&X_{12}&X_{13}\\X_{21}&X_{22}&X_{23}\\\end{matrix}\right]\in\bbM_{2,3}\otimes\bbM_{m_2,n_2}$ has Schmidt rank $m_1 n_1=6\le m_2 n_2$. Then the inequality in Theorem \ref{cj:1} is saturated if and only if $r=\rank X=\rank X_{ij}$ for $i=1,2;j=1,2,3$ up to local equivalence. Further $X$ is locally equivalent to
\begin{eqnarray}
\label{X2times3} 
\left[\begin{matrix}\left[\begin{matrix}I_r&O&O\\O&O&O\\O&O&O\\\end{matrix}\right]&\left[\begin{matrix}O&I_r&O\\O&O&O\\O&O&O\\\end{matrix}\right]&\left[\begin{matrix}O&O&I_r\\O&O&O\\O&O&O\\\end{matrix}\right]\\\left[\begin{matrix}O&O&O\\I_r&O&O\\O&O&O\\\end{matrix}\right]&\left[\begin{matrix}O&O&O\\O&I_r&O\\O&O&O\\\end{matrix}\right]&\left[\begin{matrix}O&O&O\\O&O&I_r\\O&O&O\\\end{matrix}\right]\\\end{matrix}\right].
\end{eqnarray}
Correspondingly, we have that
\begin{eqnarray}
&\rank X^\G=\rank \left[\begin{matrix}\left[\begin{matrix}I_r&O&O\\O&O&O\\O&O&O\\\end{matrix}\right]&\left[\begin{matrix}O&O&O\\I_r&O&O\\O&O&O\\\end{matrix}\right]\\\left[\begin{matrix}O&I_r&O\\O&O&O\\O&O&O\\\end{matrix}\right]&\left[\begin{matrix}O&O&O\\O&I_r&O\\O&O&O\\\end{matrix}\right]\\\left[\begin{matrix}O&O&I_r\\O&O&O\\O&O&O\\\end{matrix}\right]&\left[\begin{matrix}O&O&O\\O&O&I_r\\O&O&O\\\end{matrix}\right]\\\end{matrix}\right]=6r\\
&=6\rank \left[\begin{matrix}\left[\begin{matrix}I_r&O&O\\O&O&O\\O&O&O\\\end{matrix}\right]&\left[\begin{matrix}O&I_r&O\\O&O&O\\O&O&O\\\end{matrix}\right]&\left[\begin{matrix}O&O&I_r\\O&O&O\\O&O&O\\\end{matrix}\right]\\\left[\begin{matrix}O&O&O\\I_r&O&O\\O&O&O\\\end{matrix}\right]&\left[\begin{matrix}O&O&O\\O&I_r&O\\O&O&O\\\end{matrix}\right]&\left[\begin{matrix}O&O&O\\O&O&I_r\\O&O&O\\\end{matrix}\right]\\\end{matrix}\right]=6\rank X.
\end{eqnarray}

\section{Bipartite matrices of Schmidt rank two}
\label{sec:rank=bm}

In this section, we  transform a Schmidt-rank-two matrix into a simplified form in Sec. \ref{sec:sr2} and \ref{sec:sr2=new}, respectively using two ways. Then we study the rank of partial transpose of a bipartite matrix in Sec. \ref{sec:rank+partialtranspose}.

\subsection{transformation of Schmidt-rank-two matrices}
\label{sec:sr2}

\begin{lemma}
\label{SVDX}
Let $X$ be a Schmidt rank-two matrix $M$ in Theorem \ref{cj:1}, then $X$ can be reduced to a quasi-diagonal form, with most block matrices taking on a standard form of \eqref{SVD_result:EK}. The only imperfection is that there always remains an irreducible matrix in the upper-right corner, which depends on the specific form of the original matrix.
\end{lemma}

\begin{proof}
Let $X$ be a Schmidt rank-two matrix and $P$, $Q$ two invertible product matrices. We have $\rank X=\rank(PXQ)=\rank((P_1\otimes P_2)(A_1\otimes B_1+A_2\otimes B_2)(Q_1\otimes Q_2))=\rank(\begin{bmatrix}I_r&O\\O&O\\\end{bmatrix}\otimes \begin{bmatrix}I_s&O\\O&O\\\end{bmatrix}+\begin{bmatrix}A_{21}&A_{22}\\A_{23}&A_{24}\\\end{bmatrix}\otimes \begin{bmatrix}B_{21}&B_{22}\\B_{23}&B_{24}\\\end{bmatrix})$ where $r$ is the rank of $A_1$ and $s$ is the rank of $B_1$. Using SVD, $A_{24}$ can also be converted to a simpler form, i.e., $FA_{24}G=\left[\begin{matrix}I_t&O\\O&O\\\end{matrix}\right]$. Meanwhile, since the second row of the first block matrix is zero, we can add multiples of the row containing $A_{24}$ to the first row without altering its structure. To achieve a more concise form, we express the identity matrices as the unitary matrices. Hence, the above expression equals $\rank(\left[\begin{matrix}UU^\ast&O\\O&O\\\end{matrix}\right]\otimes \left[\begin{matrix}VV^\ast&O\\O&O\\\end{matrix}\right]+\left[\begin{matrix}U{A_{21}U}^\ast&O&UA_{22}\\O&I_t&O\\A_{23}U^\ast&O&O\\\end{matrix}\right]\otimes \left[\begin{matrix}V{B_{21}V}^\ast&O&VB_{22}\\O&I_t&O\\B_{23}V^\ast&O&O\\\end{matrix}\right])$ where we do not concern with whether $A_{ij},i,j=1,2,3,4$ in these expressions are exactly the same. Employing elementary transformations once more, we can modify the location of submatrices. So the above expression also equals 
\begin{eqnarray}
\label{eq:UU*}
\rank(\left[\begin{matrix}UU^\ast&O\\O&O\\\end{matrix}\right]\otimes \left[\begin{matrix}VV^\ast&O\\O&O\\\end{matrix}\right]+\left[\begin{matrix}O&A_{11}&O&A_{12}\\O&O&I_t&O\\I_k&O&O&O\\\end{matrix}\right]\otimes \left[\begin{matrix}O&B_{11}&O&B_{12}\\O&O&I_g&O\\I_w&O&O&O\\\end{matrix}\right]),
\end{eqnarray}
where the subscript of the identity matrix continues to represent the rank of the associated matrix. Here, $UU^\ast$ and $VV^*$ are identity matrices of order $r$ and $s$, respectively. Further, $\left[\begin{matrix}O&A_{11}\\O&O\\\end{matrix}\right]$ and $\left[\begin{matrix}O&B_{11}\\O&O\\\end{matrix}\right]$ are matrices of order $r$ and $s$, respectively. One can see that $A_{11}$ is of size $(r-t)\times(r-k)$ and $B_{11}$ is of size $(s-g)\times(s-w)$.
So we obtain that $A_{12}$ is of size $(r-t)\times(n_1-r-t)$, $B_{12}$ is of size $(s-g)\times(n_2-s-g)$, and
\begin{eqnarray}
r+k=m_1,    
\quad
s+w=m_2.
\end{eqnarray}
Using the above facts, we can rewrite \eqref{eq:UU*} as
\begin{eqnarray}
\label{eq:UU*-1}
\rank(
\left[\begin{matrix}I_r&O\\O&O\\\end{matrix}\right]
\otimes \left[\begin{matrix}I_s&O\\O&O\\\end{matrix}\right]
+
\left[\begin{matrix}O&A_{11}&O&A_{12}\\O&O&I_t&O\\I_{m_1-r}&O&O&O\\\end{matrix}\right]\otimes \left[\begin{matrix}O&B_{11}&O&B_{12}\\O&O&I_g&O\\I_{m_2-s}&O&O&O\\\end{matrix}\right]).    
\end{eqnarray}
Now, the first two canonical form matrices can be reformulated to match the size of the following matrices. Therefore it is also $\rank(\left[\begin{matrix}\left[\begin{matrix}I_{r-t}&O\\O&I_t\\\end{matrix}\right]&O\\O&O\\\end{matrix}\right]\otimes \left[\begin{matrix}\left[\begin{matrix}I_{s-g}&O\\O&I_g\\\end{matrix}\right]&O\\O&O\\\end{matrix}\right]+\left[\begin{matrix}\left[\begin{matrix}O&A_{11}\\O&O\\\end{matrix}\right]&\left[\begin{matrix}O&A_{12}\\I_t&O\\\end{matrix}\right]\\\begin{matrix}I_k&O\\\end{matrix}&\begin{matrix}O&O\\\end{matrix}\\\end{matrix}\right]\otimes \left[\begin{matrix}\left[\begin{matrix}O&B_{11}\\O&O\\\end{matrix}\right]&\left[\begin{matrix}O&B_{12}\\I_g&O\\\end{matrix}\right]\\\begin{matrix}I_w&O\\\end{matrix}&\begin{matrix}O&O\\\end{matrix}\\\end{matrix}\right])$. The following transformation can be applied to the identity matrix, with corresponding transformations performed on the two matrices later. We have $UI_rU^\ast=\left[\begin{matrix}V&O\\O&W\\\end{matrix}\right]I_r\left[\begin{matrix}V^\ast&O\\O&W^\ast\\\end{matrix}\right]$ where $U=\left(\begin{matrix}V&O\\O&W\\\end{matrix}\right)$. Respectively, we can apply transformations to the matrix by left-multiplying $\left[\begin{matrix}O&A_{11}\\O&O\\\end{matrix}\right]$ with $\left[\begin{matrix}V&O\\O&W\\\end{matrix}\right]$ and right-multiplying it with $U=\left[\begin{matrix}V^\ast&O\\O&W^\ast\\\end{matrix}\right]$. 
This is equivalent to left-multiplying the first $(r-t)$ rows of $\left[\begin{matrix}O&A_{11}\\O&O\\\end{matrix}\right]$ by $V$ and right-multiplying its last $t$ columns by $W$. Thus, $VA_{11}W^*$ is in the upper-right corner of the original matrix $\left[\begin{matrix}O&A_{11}\\O&O\\\end{matrix}\right]$. In addition, we assume that $VA_{11}W^*$ can transform $A$ to a diagonal matrix $\left[\begin{matrix}d_1&0&\cdots&0&\cdots&0\\0&d_2&\cdots&0&\cdots&0\\\vdots&\vdots&\ddots&\vdots&\cdots&\vdots\\0&0&\cdots&d_h&\cdots&0\\\vdots&\vdots&\vdots&\vdots&\ddots&\vdots\\0&0&\cdots&0&\cdots&0\\\end{matrix}\right]$ where $h$ is $\rank A$. When we left-multiply $\left[\begin{matrix}O&{VA}_{11}W^\ast\\O&O\\\end{matrix}\right]$ with $\left[\begin{matrix}{d_1}^{-1}&0&\cdots&0&0&\cdots&0\\0&{d_2}^{-1}&\cdots&0&0&\cdots&0\\\vdots&\vdots&\ddots&\vdots&\vdots&\cdots&\vdots\\0&0&\cdots&{d_h}^{-1}&0&\cdots&0\\0&0&\cdots&0&1&\cdots&0\\\vdots&\vdots&\vdots&\vdots&\vdots&\ddots&\vdots\\0&0&\cdots&0&0&\cdots&1\\\end{matrix}\right]$, the original matrix has a more concise form. $\left[\begin{matrix}O&A_{11}&O&A_{12}\\O&O&I_t&O\\I_{m_1-r}&O&O&O\\\end{matrix}\right]$ can be transformed to $\left[\begin{matrix}\left[\begin{matrix}O&\left[\begin{matrix}I_h&O\\O&O\\\end{matrix}\right]\\O&O\\\end{matrix}\right]&\left[\begin{matrix}O&A_{12}\\I_t&O\\\end{matrix}\right]\\\begin{matrix}I_{m_1-r}&O\\\end{matrix}&\begin{matrix}O&O\\\end{matrix}\\\end{matrix}\right]$. 

Similarly, for the other half of the original expression, we can apply transformations using unitary matrices. Let $j$ be the rank of matrix $B_{11}$. Then $\left[\begin{matrix}O&B_{11}&O&B_{12}\\O&O&I_g&O\\I_{m_2-s}&O&O&O\\\end{matrix}\right]$ can be transformed into $\left[\begin{matrix}\left[\begin{matrix}O&\left[\begin{matrix}I_j&O\\O&O\\\end{matrix}\right]\\O&O\\\end{matrix}\right]&\left[\begin{matrix}O&B_{12}\\I_g&O\\\end{matrix}\right]\\\begin{matrix}I_{m_2-s}&O\\\end{matrix}&\begin{matrix}O&O\\\end{matrix}\\\end{matrix}\right]$.\\

Therefore the object we want to research has been transformed into the $\rank$ of $\left[\begin{matrix}\left[\begin{matrix}I_{r-t}&O\\O&I_t\\\end{matrix}\right]&O\\O&O\\\end{matrix}\right]\otimes \left[\begin{matrix}\left[\begin{matrix}I_{s-g}&O\\O&I_g\\\end{matrix}\right]&O\\O&O\\\end{matrix}\right]+\left[\begin{matrix}\left[\begin{matrix}O&\left[\begin{matrix}I_h&O\\O&O\\\end{matrix}\right]\\O&O\\\end{matrix}\right]&\left[\begin{matrix}O&A_{12}\\I_t&O\\\end{matrix}\right]\\\begin{matrix}I_{m_1-r}&O\\\end{matrix}&\begin{matrix}O&O\\\end{matrix}\\\end{matrix}\right]\otimes \left[\begin{matrix}\left[\begin{matrix}O&\left[\begin{matrix}I_j&O\\O&O\\\end{matrix}\right]\\O&O\\\end{matrix}\right]&\left[\begin{matrix}O&B_{12}\\I_g&O\\\end{matrix}\right]\\\begin{matrix}I_{m_2-s}&O\\\end{matrix}&\begin{matrix}O&O\\\end{matrix}\\\end{matrix}\right]$. \\

This is equivalent to 

\begin{eqnarray}
\label{SVD_result:EK}
\left[\begin{matrix}\left[\begin{matrix}E&O&\cdots&O\\O&E&\cdots&O\\\vdots&\vdots&\ddots&\vdots\\O&O&\cdots&E\\\end{matrix}\right]&O\\O&O\\\end{matrix}\right]+\left[\begin{matrix}\left[\begin{matrix}O&\left[\begin{matrix}\left[\begin{matrix}K&O&\cdots&O\\O&K&\cdots&O\\\vdots&\vdots&\ddots&\vdots\\O&O&\cdots&K\\\end{matrix}\right]&O\\O&O\\\end{matrix}\right]\\O&O\\\end{matrix}\right]&\left[\begin{matrix}O&A_{12}\\I_t&O\\\end{matrix}\right]\\\begin{matrix}I_{m_1-r}&\ \ \ \ \ \ \ \ \ \ \ \ \ \ \ \
O\\\end{matrix}&\begin{matrix}O&O\\\end{matrix}\\\end{matrix}\right]. 
\end{eqnarray}

where $K=\left[\begin{matrix}\left[\begin{matrix}O&\left[\begin{matrix}I_j&O\\O&O\\\end{matrix}\right]\\O&O\\\end{matrix}\right]&\left[\begin{matrix}O&B_{12}\\I_g&O\\\end{matrix}\right]\\\begin{matrix}I_{m_2-s}&O\\\end{matrix}&\begin{matrix}O&O\\\end{matrix}\\\end{matrix}\right]$ and $\left[\begin{matrix}K&O&\cdots&O\\O&K&\cdots&O\\\vdots&\vdots&\ddots&\vdots\\O&O&\cdots&K\\\end{matrix}\right]$ is a $h\times h$ matrix, $E=\left[\begin{matrix}I_g&O\\O&O\\\end{matrix}\right]$ and $\left[\begin{matrix}E&O&\cdots&O\\O&E&\cdots&O\\\vdots&\vdots&\ddots&\vdots\\O&O&\cdots&E\\\end{matrix}\right]$ is a $r\times r$ matrix. 
\end{proof}

Here, we can repeatedly use SVD to simplify the matrix, reducing it to a quasi-diagonal form, with most block matrices taking on a standard form. The only imperfection is that there always remains an irreducible matrix in the upper-right corner, which depends on the specific form of the original matrix.

\subsection{the range and span related to Schmidt-rank-two matrices}
\label{sec:sr2=new}

\begin{lemma}
\label{le:X=AotimesB+CotimesD}
Let $X=A\otimes B+C\otimes D \in \bbM_d\otimes\bbM_f$. 

(i) If the range of $A$ and $C$ has a full-rank matrix then up to local equivalence, we can assume that $A$ is the order-$d$ identity matrix.   

(ii) If additionally, the range of $B$ and $D$ has also a full-rank matrix then up to local equivalence, we can assume that $A$ and $D$ are respectively the order-$d$ and order-$f$ identity matrix.   

(iii) Using the same hypothesis of claim (ii), we can also assume that $A$ and $B$ are respectively the order-$d$ and order-$f$ identity matrix. 
\end{lemma}
\begin{proof}
(i) Using the hypothesis, we can assume that $A+xC$ is a full-rank matrix with a complex number $x$. Hence
\begin{eqnarray}
X=(A+xC)\otimes B+
C\otimes(-xB+D).
\end{eqnarray}
Using local equivalence, we can assume that $A+xC$ is the identity matrix. So claim (i) holds.

(ii) and (iii) The claim can be proven by following the proof of (i) and Lemma \ref{le:X=A+yB}. We have completed the proof.
\end{proof}

Using the above lemma, one can readily show the following observation.
\begin{corollary}
\label{cr:X=AotimesB+CotimesD}
Let $X=A\otimes B+C\otimes D \in \bbM_d\otimes\bbM_f$. 

(i) Suppose the span of $A$ and $C$ has a full-rank matrix, and the range of $B$ and $D$ has also a full-rank matrix. Then we can assume that $A$ and $D$ are full-rank matrices, and the eigenvalues of $(A^{-1}\otimes I_f)X(I_d\otimes D^{-1})$ are $\a_j+\b_k$, $j=1,...,d$ and $k=1,...,f$, where $\a_j$'s and $\b_j$'s are respectively the eigenvalues of $A^{-1}C$ and $BD^{-1}$. 
%As a result, the rank of $X$ equals the number of nonzero $(\a_j+\b_k)'$s. 

(ii) Suppose the span of $A$ and $C$ has a full-rank matrix. Then we can assume that $A$ is a full-rank matrix, and the eigenvalues of $(A^{-1}\otimes I_f)X$ are the union of those of $B+\g_jD$, where $j=1,...,d$ and $\g_j$'s are the eigenvalues of $A^{-1}C$. 
%As a result, the rank of $X$ equals the number of nonzero elements in the union. 
\qed
\end{corollary}

Following the above argument, one can obtain a general result of the rank of Schmidt-rank-two $X$.
\begin{lemma}
\label{le:X=AotimesB+CotimesD=square}
Let $X=A\otimes B+C\otimes D \in \bbM_d\otimes\bbM_f$. If $A$ and $C$ are simultaneously diagonalizable and have eigenvalues $a_1,...a_d$ and $c_1,...,c_d$, respectively, then the eigenvalues of $X$ are the union of those of $a_jB+c_jD$ for $j=1,...,d$. The rank of $X$ equals the sum of those of $a_jB+c_jD$ for $j=1,...,d$.   
\end{lemma}
\begin{proof}
Because $A$ and $C$ are diagonalizable simultaneously, we can assume that $P^{-1}AP=\left[\begin{matrix}a_1&\cdots&0\\\vdots&\ddots&\vdots\\0&\cdots&a_d\\\end{matrix}\right]$ and $P^{-1}CP=\left[\begin{matrix}c_1&\cdots&0\\\vdots&\ddots&\vdots\\0&\cdots&c_d\\\end{matrix}\right]$. We perform the following matrix multiplication, 
\begin{eqnarray}
&&(P\otimes I_f)^{-1}X(P\otimes I_f)\\
=&&(P\otimes I_f)^{-1}[P \left[\begin{matrix}a_1&\cdots&0\\\vdots&\ddots&\vdots\\0&\cdots&a_d\\\end{matrix}\right] P^{-1}\otimes B+P\left[\begin{matrix}c_1&\cdots&0\\\vdots&\ddots&\vdots\\0&\cdots&c_d\\\end{matrix}\right] P^{-1}\otimes D ](P\otimes I_f)\notag\\
=&&
\left[\begin{matrix}a_1&\cdots&0\\\vdots&\ddots&\vdots\\0&\cdots&a_d\\\end{matrix}\right] \otimes B+\left[\begin{matrix}c_1&\cdots&0\\\vdots&\ddots&\vdots\\0&\cdots&c_d\\\end{matrix}\right] \otimes D\\
=&&
\left[\begin{matrix}a_1B+c_1D&\cdots&0\\\vdots&\ddots&\vdots\\0&\cdots&a_dB+c_dD\\\end{matrix}\right].
\end{eqnarray}
\end{proof}

For the case of non-square matrices $A,B,C,$ and $D$, we can append zero rows or columns so that they become square matrices. Then Lemma \ref{le:X=AotimesB+CotimesD=square} applies.

%\opp What if the precondition of Lemma \ref{le:X=AotimesB+CotimesD=square} is not satisfied? 
%\begin{eqnarray}
%X=&&A\otimes B+C\otimes D\\=&&\left[\begin{matrix}{A_1}^{m\times n}&A_2\\O&A_3\\\end{matrix}\right]\otimes B+\left[\begin{matrix}{C_1}^{m\times n}&C_2\\O&C_3\\\end{matrix}\right]\otimes D,
%\end{eqnarray}
%where $m< n$.

\subsection{rank of matrix and its partial transpose}

\label{sec:rank+partialtranspose}

Based on the above discussion, we can explore more types of matrices that satisfies $\rank(X)= \rank (X^{\Gamma})$ where $\Gamma$ is the partial transpose w.r.t some of the systems in $X$. 

(i) To begin with, let us consider the first scenario. If $A$ is a diagonal matrix, we assume $A=diag{(\lambda_1,\lambda_2\ldots\lambda_n)}$ where $\lambda_i\neq0\ (i=1,2,\ldots\ n)$. As is mentioned before, for $X=A\otimes B$, $\rank(X)= \rank (X^\Gamma)$ where $\Gamma$ is the partial transpose w.r.t system $A$. That is because in fact, the rank of the transpose of $A$ equals the $rank$ of $A$. Matrix $X$ is equivalent to $B$ matrix with non-zero numbers multiplied and placed along the diagonal. Also, $\rank(X)= \rank (X^\Gamma)$ where $\Gamma$ is the partial transpose w.r.t system $B$ because the rank of the transpose of $B$ equals the rank of $B$.

(ii) if A is in the standard form, we assume $A=\left[\begin{matrix}\lambda_1&0&\cdots&0&0&\cdots&0\\0&\lambda_2&\cdots&0&0&\cdots&0\\\vdots&\vdots&\ddots&\vdots&\vdots&\vdots&\vdots\\0&0&\cdots&\lambda_r&0&\cdots&0\\0&0&0&0&0&\cdots&0\\\vdots&\vdots&\vdots&\vdots&\vdots&\ddots&\vdots\\0&0&\cdots&0&0&\cdots&0\\\end{matrix}\right]$ where $A$ is an $m\times n$ matrix. Whether $\Gamma$ is the partial transpose w.r.t system $A$ or $\Gamma$ is the partial transpose w.r.t system $B$, we can both arrive at the conclusion that $\rank(X)= \rank (X^\Gamma)$. 

(iii) Thirdly, more generally, if $A$ is a block diagonal matrix, we assume $A=\left[\begin{matrix}A_1&O&\cdots&O\\O&A_2&\cdots&O\\\vdots&\vdots&\ddots&\vdots\\O&O&\cdots&A_n\\\end{matrix}\right]$. Then we can transform the problem into discussing each block of $A_i(i=1,2,\ldots\ n)$. Based on the above discussion, we know that for $X_i=\ A_i\otimes\ B$, in which $A_i\otimes\ B$ means the Kronecker product of matrix $A_i$ and $B$ and for matrix $X_i^\Gamma$, where $\Gamma$ is the partial transpose w.r.t system $A_i$, $B$ or $A_i$ and $B$, $\rank(X_i)= \rank (X_i^\Gamma)$. Besides the rank of the large matrix $A$ is the sum of the ranks of its $n$ smaller submatrices due to the special positions of these  smaller submatrices. That is why $\rank(X)= \rank (X^\Gamma)$ in this scenario, too. So we present the first fact of this paper.

\begin{lemma}
\label{le:blockDiagonal}
A block-diagonal matrix $X$ satisfies $\rank X= \rank X^\G$.
\qed
\end{lemma}
Because a generic matrix has full rank, one can say that the equality in Lemma \ref{le:blockDiagonal} holds generically. However, constructing a specific family of matrices saturating the equality is a challenge. 
An idea is to try Schmidt-rank-two $X$ by using the previously established results.

\section{Conclusions}
\label{sec:con}

We have shown that  matrix inequality $\rank(\sum^K_{j=1} A_j^T \otimes B_j)\le K \cdot \rank(\sum^K_{j=1} A_j \otimes B_j)$ can be saturated when $A_i$'s are column or row vectors, or two-by-two matrices. We also study the general case for $K=2$. An open problem arising from this paper is to find out the condition saturating the inequality for Schmidt-rank-two matrices with any dimensions. One may be interested in extending the above results to bipartite matrices of larger Schmidt rank. Another issue is fully construct the matrices $X$ such that $\rank(X)= \rank (X^{\Gamma})$.

As the next step, we propose a conjecture related to the above facts, based on Theorem \ref{le:ANYm2,K>2} and equation \eqref{XEij}.

\section*{Acknowledgments}
\label{sec:ack}	

%We thank XXX for valuable discussions. 

Authors were supported by the NNSF of China (Grant No. 12471427), and the Fundamental Research Funds for the Central Universities (Grant Nos. ZG216S2110).

\bibliographystyle{unsrt}

\bibliography{Thefinalanswer} 

\begin{thebibliography}{10}

\bibitem{hhh96}
M.~{Horodecki}, P.~{Horodecki}, and R.~{Horodecki}.
\newblock {Separability of mixed states: necessary and sufficient conditions}.
\newblock {\em Physics Letters A}, 223:1--8, February 1996.

\bibitem{horodecki1997}
P.~Horodecki.
\newblock Separability criterion and inseparable mixed states with positive
  partial transposition.
\newblock {\em Phys. Lett. A}, 232:333, 1997.

\bibitem{hh1999}
M.~Horodecki and P.~Horodecki.
\newblock Reduction criterion of separability and limits for a class of
  distillation protocols.
\newblock {\em Phys. Rev. A}, 59:4206, 1999.

\bibitem{Divincenzo2000Evidence}
D.P. Divincenzo, P.W. Shor, J.A. Smolin, B.M. Terhal, and A.V. Thapliyal.
\newblock Evidence for bound entangled states with negative partial transpose.
\newblock {\em Physical Review A}, 61(6):200--200, 2000.

\bibitem{Chen2008Rank}
L.~Chen and Y.X. Chen.
\newblock Rank three bipartite entangled states are distillable.
\newblock {\em Physical Review A}, 78(2):3674--3690, 2008.

\bibitem{Chen2012Distillability}
L.~Chen and D.\v{Z}. Djokovi\'{c}.
\newblock Distillability and ppt entanglement of low-rank quantum states.
\newblock {\em Journal of Physics A Mathematical and Theoretical},
  44(28):1213--1219, 2012.

\bibitem{Lin2016Non}
L.~Chen and D.\v{Z}. Djokovi\'{c}.
\newblock Non-positive-partial-transpose quantum states of rank four are
  distillable.
\newblock 2016.

\bibitem{QIAN2021139}
L.L. Qian, L.~Chen, D.L. Chu, and Y.~Shen.
\newblock A matrix inequality for entanglement distillation problem.
\newblock {\em Linear Algebra and its Applications}, 616:139--177, 2021.

\bibitem{Eisert2000The}
J.~Eisert and H.J. Briegel.
\newblock The schmidt measure as a tool for quantifying multi-particle
  entanglement.
\newblock {\em Physical Review A}, 64(2):17--18, 2000.

\bibitem{2005A}
N.~Linden and A.~Winter.
\newblock A new inequality for the von neumann entropy.
\newblock {\em Communications in Mathematical Physics}, 259(1):129--138, 2005.

\bibitem{2012Infinitely}
J.~Cadney, N.~Linden, and A.~Winter.
\newblock Infinitely many constrained inequalities for the von neumann entropy.
\newblock {\em IEEE Transactions on Information Theory}, 58(6):3657--3663,
  2012.

\bibitem{cadney2014}
Josh Cadney, Marcus Huber, Noah Linden, and Andreas Winter.
\newblock Inequalities for the ranks of multipartite quantum states.
\newblock {\em LINEAR ALGEBRA AND ITS APPLICATIONS}, 452:153--171, JUL 1 2014.

\bibitem{WOS:000966945300001}
Zhiwei Song, Lin Chen, Yize Sun, and Mengyao Hu.
\newblock Proof of a conjectured 0-renyi entropy inequality with applications
  to multipartite entanglement.
\newblock {\em IEEE TRANSACTIONS ON INFORMATION THEORY}, 69(4):2385--2399, APR
  2023.

\bibitem{tan}
Ziyue Tan.
\newblock Equality condition for a matrix inequality with kronecker product.
\newblock {\em INTERNATIONAL JOURNAL OF QUANTUM INFORMATION}, 2025 FEB 27 2025.

\bibitem{WOS:000222471400029}
M~Hein, J~Eisert, and HJ~Briegel.
\newblock Multiparty entanglement in graph states.
\newblock {\em PHYSICAL REVIEW A}, 69(6), JUN 2004.

\bibitem{WOS:000187004700026}
TC~Wei and PM~Goldbart.
\newblock Geometric measure of entanglement and applications to bipartite and
  multipartite quantum states.
\newblock {\em PHYSICAL REVIEW A}, 68(4), OCT 2003.

\bibitem{WOS:000463910100001}
Eric Chitambar and Gilad Gour.
\newblock Quantum resource theories.
\newblock {\em REVIEWS OF MODERN PHYSICS}, 91(2), APR 4 2019.

\bibitem{WOS:000274002900006}
B.~Kraus.
\newblock Local unitary equivalence of multipartite pure states.
\newblock {\em PHYSICAL REVIEW LETTERS}, 104(2), JAN 15 2010.

\bibitem{WOS:000259793800006}
Eric Chitambar, Runyao Duan, and Yaoyun Shi.
\newblock Tripartite entanglement transformations and tensor rank.
\newblock {\em PHYSICAL REVIEW LETTERS}, 101(14), OCT 3 2008.

\end{thebibliography}

\end{document}